\documentclass[a4paper,11pt,reqno]{amsart}
\usepackage[utf8]{inputenc}
\usepackage{mathrsfs}
\usepackage{dsfont}
\usepackage{hyperref}
\usepackage{amsmath}
\usepackage{amssymb}
\usepackage{amsthm}
\usepackage{amsfonts}
\usepackage{amstext}
\usepackage{amsopn}
\usepackage{amsxtra}
\usepackage{mathrsfs}
\usepackage{dsfont}
\usepackage{color}
\usepackage{esint}
\usepackage{graphicx}

\usepackage[backend=biber, style=alphabetic, natbib, labelnumber, defernumbers=true, language=english, autolang=other,maxbibnames=20,doi=false,isbn=false,url=false,eprint=false,bibencoding=utf8, firstinits=true]{biblatex}
\newtheorem{theorem}{Theorem}

\newtheorem{lemma}[theorem]{Lemma}
\newtheorem{proposition}[theorem]{Proposition}
\newtheorem{corollary}[theorem]{Corollary}
\newtheorem{remark}[theorem]{Remark}

\newcommand{\R}{\mathbb{R}}
\newcommand{\Z}{\mathbb{Z}}
\newcommand{\C}{\mathbb{C}}
\newcommand{\N}{\mathbb{N}}

\newcommand{\bS}{\mathbb{S}}

\newcommand{\ii}{\infty}
\newcommand\1{{\ensuremath {\mathds 1} }}
\renewcommand\phi{\varphi}

\newcommand{\wto}{\rightharpoonup}

\newcommand{\cZ}{\mathcal{Z}}

\newcommand{\cP}{\mathcal{P}}

\newcommand{\cB}{\mathcal{B}}

\newcommand{\cE}{\mathcal{E}}

\newcommand{\cH}{\mathcal{H}}
\newcommand{\cL}{\mathcal{L}}

\newcommand{\Per}{{\rm Per}\,}
\renewcommand{\geq}{\geqslant}
\renewcommand{\leq}{\leqslant}

\renewcommand{\tilde}{\widetilde}
\newcommand{\eps}{\varepsilon}

\newcommand{\nn}{\nonumber}
\newcommand{\rd}{\mathrm{d}}
\newcommand{\dx}{\rd x}
\newcommand{\dy}{\rd y}

\title[Liquid Drop Model in the Low Density Limit]{Liquid Drop Model for Nuclear Matter\\ in the Low Density Limit}

\date{\today.}

\author[R.L. Frank]{Rupert L. Frank}
\address[Rupert L. Frank]{Mathematisches Institut, Ludwig-Maximilans Universit\"at M\"unchen, Theresienstr. 39, 80333 M\"unchen, Germany, and Munich Center for Quantum Science and Technology (MCQST), Schellingstr. 4, 80799 M\"unchen, Germany}
\email{r.frank@lmu.de}

\author[M. Lewin]{Mathieu Lewin}
\address[Mathieu Lewin]{CEREMADE, CNRS, Université Paris-Dauphine, PSL Research University, Place de Lattre de Tassigny, 75016 Paris, France}
\email{mathieu.lewin@math.cnrs.fr}

\author[R. Seiringer]{Robert Seiringer}
\address[Robert Seiringer]{IST Austria (Institute of Science and Technology Austria), Am Campus 1, 3400 Klosterneuburg, Austria}
\email{robert.seiringer@ist.ac.at}

\begin{document}

\begin{abstract}
We consider the liquid drop model with a positive background density in the thermodynamic limit. We prove a two-term asymptotics for the ground state energy per unit volume in the dilute limit. Our proof justifies the expectation that optimal configurations consist of droplets of unit size that arrange themselves according to minimizers for the Jellium problem for point particles. In particular, we provide the first rigorous derivation of what is known as the gnocchi phase in astrophysics.

\bigskip

\noindent \sl \copyright~2025 by the authors. This paper may be reproduced, in its entirety, for non-commercial purposes.
\end{abstract}

\maketitle

\section{Introduction and main result}
In this paper we provide the first rigorous justification of the \emph{``gnocchi phase''}, which is one of the nuclear pasta phases that are believed to form in the outer crust of neutron stars. We work with the simplest physical model used to describe this situation, called the \emph{liquid drop}. It is due to Gamow~\cite{Gamow-29,Gamow-30} and was later refined by Bohr and Wheeler~\cite{BohWhe-39}. The same model can also describe microphase separation of diblock copolymers~\cite{OhtKaw-86}.

From a broader perspective, this is an example of a pattern formation due to the competition of attractive geometric short range effects and long range repulsive forces.

Let $\Lambda\subset\R^3$ be a bounded open set which represents a piece of the external crust of a neutron star. We work at the microscopic scale and should therefore think that $\Lambda$ is really huge. Later we will in fact take the limit $\Lambda\nearrow\R^3$. In the outer part of the crust of a neutron star the pressure is large enough that most of the electrons are pulled out of the atoms, but still not strong enough to dissociate the nuclei into protons and neutrons. In the regime of interest for this work, the nuclei will form aggregates behaving like liquid droplets and evolving in a uniform sea of electrons. Assuming that the nuclei form a constant density phase over an unknown Borel set $\Omega\subset\Lambda$, the liquid drop model postulates that the energy is given by
\begin{equation}
\boxed{\cE_\Lambda[\rho,\Omega]:=\Per(\Omega)+\frac12\iint_{\R^3\times\R^3}\frac{(\1_\Omega-\rho\1_{\Lambda})(x)(\1_\Omega-\rho\1_{\Lambda})(y)}{|x-y|}\,\dx\,\dy.}
\label{eq:def_energy}
\end{equation}
Here $\Per(\Omega)$ denotes the surface area of a smooth set $\Omega$. For general measurable sets it is understood as the perimeter in the sense of geometric measure theory, see~\cite{Maggi-12} for a textbook introduction. The perimeter takes into account the short range interactions between the nuclei in a droplet. The second term is the electrostatic energy due to the fact that electrons and protons are, respectively, negatively and positively charged. The constant $\rho$ is the relative fraction of delocalized electrons. We have chosen units so that the local nuclear density and the constant in front of the perimeter are both equal to 1. In physical units they depend on the proton-neutron ratio. Note that protons are not able to bind without neutrons.

If we swap the roles of $\Omega$ and $\Lambda\setminus\Omega$, we find
\begin{equation}
\cE_\Lambda[\rho,\Lambda\setminus \Omega]=\Per(\Lambda)+\cE_\Lambda[1-\rho,\Omega].
\label{eq:symmetry_rho}
\end{equation}
In a neutron star, matter is essentially neutral,
$$|\Omega|\approx\rho|\Lambda|.$$
Since $\Omega\subset\Lambda$, this implies $0\leq \rho\leq1$ and then by symmetry we can assume $0\leq \rho\leq 1/2$.

At equilibrium the nuclear phase $\Omega\subset\Lambda$ minimizes the energy functional in~\eqref{eq:def_energy}. Working under the constraint that the system is exactly neutral, the minimum equals
\begin{equation}
E_\Lambda(\rho):=\min_{\substack{\Omega\subset\Lambda\\ |\Omega|=\rho|\Lambda|}}\cE_\Lambda[\rho,\Omega].
\label{eq:def_E}
\end{equation}
A natural question is to study the shape of these minimizers in the low density regime. As a first step in this direction, we will study the resulting energy asymtotics in this regime. We will first take $\Lambda\nearrow\R^3$ before we actually study the limit $\rho\to0$. This order of the limits is important for the result to be physically meaningful.

\subsection{Thermodynamic limit}
We first look at the limit of an infinite piece of matter, which is justified since the model is microscopic.  The following result states that the energy per unit volume converges to a limit independent of the shape of $\Lambda$.

\begin{theorem}[Thermodynamic limit]\label{thm:thermo_limit}
Let $0\leq\rho\leq1$. Let $\Lambda_n$ be a sequence of $(r,L)$--Lipschitz open sets, so that
\begin{equation}
B(0,\ell_n/C)\subset \Lambda_n\subset B(0,\ell_n)
\label{eq:ass_balls_Lambda_n}
\end{equation}
for some $\ell_n\to\ii$ and some $C$. We also assume its boundary satisfies the Fisher regularity condition
\begin{equation}
\big|\partial\Lambda_n+B(0,h)\big|\leq C\ell_n^2 h,\qquad \forall h\leq \ell_n.
\label{eq:Fisher}
\end{equation}
Then the following limit
\begin{equation}
\boxed{e(\rho):=\lim_{n\to\ii}\frac{E_{\Lambda_n}(\rho)}{|\Lambda_n|}}
\label{eq:thermo_limit}
\end{equation}
exists and does not depend on the sequence $\Lambda_n$, nor on the constants $r,L,C$. The resulting energy is symmetric about $\rho=1/2$: $e(\rho)=e(1-\rho)$.
\end{theorem}

For us an open set $\Lambda\subset\R^3$ is $(r,L)$--Lipschitz  if for any $x\in\partial\Lambda$, the set $\Lambda\cap B(x,r)$ is, in an appropriate coordinate system, equal to the epigraph of an $L$--Lipschitz function. Under this assumption, one can see that the Fisher condition~\eqref{eq:Fisher} holds for $h\leq Cr$. Hence~\eqref{eq:Fisher} is only an assumption for large $h$. Note that all the conditions are satisfied if $\Lambda_n$ is the rescaling of a fixed $(r,L)$--Lipschitz bounded open set containing the origin (for instance a cube or a ball). Our assumptions on the domain $\Lambda_n$ are not sharp and have been chosen because they are relatively simple.

Before we provide more comments on Theorem~\ref{thm:thermo_limit}, let us quickly describe the periodic version of the problem, for a better comparison with the existing literature. We take $\Lambda_L=(-L/2,L/2)^3$ a cube of side length $L$ and let $G_L$ denote the $L\Z^3$-periodic Coulomb potential. The Fourier coefficients of $G_L$ coincide with the Fourier transform of $1/|x|$ on the grid $(2\pi/L)\Z^3$. It is also proportional to the Green's function of the Laplacian with periodic boundary conditions. The zero mode $\widehat{G_L}(0)$ is taken equal to zero for simplicity. It is then known that the limit is the same as~\eqref{eq:thermo_limit}:
\begin{equation}
e(\rho)=\lim_{L\to\ii}\frac1{L^3}\min_{\substack{\Omega\subset\Lambda_L\\ |\Omega|=\rho L^3}}\left\{\Per(\Omega)+\frac12\iint_{\Omega^2}G_L(x-y)\dx\,\dy\right\}.
\label{eq:thermo_limit_periodic}
\end{equation}
Note that since we have fixed $\int_{\Lambda_L}G_L=0$, the electronic background disappears from the energy, but it is still present in the constraint that $|\Omega|=\rho L^3$.

The existence of the limit~\eqref{eq:thermo_limit_periodic} for the periodic model is a consequence of the seminal work~\cite{AlbChoOtt-09}, where the optimal convergence rate was obtained as well as information about the minimizers. That this also implies that the limit in~\eqref{eq:thermo_limit} exists for cubes and coincides with~\eqref{eq:thermo_limit_periodic} was noted in~\cite{EmmFraKon-20}. In Appendix~\ref{app:proof_thermo-limit} we provide a completely different proof of Theorem~\ref{thm:thermo_limit} and of the limit~\eqref{eq:thermo_limit_periodic} using classical methods for Coulomb systems introduced by Lieb, Lebowitz and Narnhofer in the 1970s~\cite{LieLeb-72,LieNar-75}. This proof works for any sequence $\Lambda_n$ satisfying the mentioned conditions.

\subsection{Low density and pasta phases.}
Next we consider the limit $\rho\to0$ for the energy per unit volume $e(\rho)$ of an infinite piece of nuclear matter with average density $\rho$ of electrons and nuclear droplets. We state our theorem right now, but the constants appearing in the statement will be defined a bit later. This is the \textbf{main result of the paper}.

\begin{theorem}[Low density]\label{thm:low_density}
In the limit $\rho\to0$, we have the expansion
\begin{equation}
\boxed{e(\rho)=\mu_*\,\rho+m_*^{\frac23}\,e_{\rm Jel}\,\rho^{\frac43}+o(\rho^{\frac43})}
\label{eq:low_density_expansion}
\end{equation}
where $\mu_*$ and $m_*$ are the energy per unit volume and mass of an isolated nuclear droplet defined below in~\eqref{eq:def_mu_*} and~\eqref{eq:def_m_*}, whereas $e_{\rm Jel}<0$  is the Jellium energy defined below in~\eqref{eq:Jellium}
\end{theorem}

The theorem provides an expansion of the liquid drop energy with the different orders in $\rho$ corresponding to different properties of minimizers. The leading term $\mu_*\rho$ describes the fact that minimizers $\Omega$ will consist of infinitely many droplets of size of order one (independent of $\rho$), each droplet minimizing the liquid drop energy~\eqref{eq:def_E} without any background ($\rho=0$) normalized by the volume.  It is conjectured that these droplets are balls, in which case we have $R_*=(\frac{15}{8\pi})^{1/3}$, $\mu_*=9(\frac\pi{15})^{1/3}$ and $m_*=5/2$. Those represent the \emph{``gnocchis''}. An illustration is provided in Figure~\ref{fig:big-box}.

The next term in the expansion~\eqref{eq:low_density_expansion} describes how the gnocchis are arranged in the tomato sauce, corresponding to the low-density uniform background of electrons. The Jellium model (introduced below) tells us that the gnocchis will be placed at large distances, of order $\rho^{-1/3}$. A famous conjecture of Wigner~\cite{Wigner-34} states that they will in fact be regularly arranged on a Body-Centered Cubic lattice, in which case the Jellium energy is $e_{\rm Jel}\simeq -1.4442$. The interpretation of the negativity of $e_{\rm Jel}$ is that there is an effective attraction between the droplets, mediated by the background of electrons.

\begin{figure}[t]
\includegraphics[width=7cm]{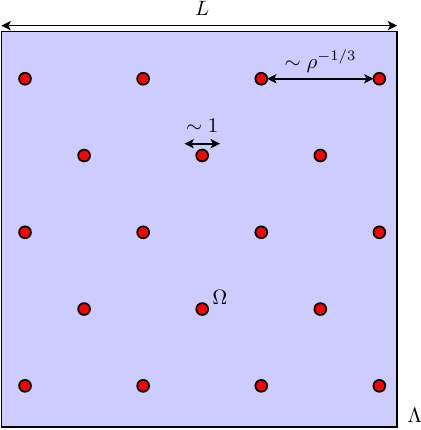}
\caption{Schematic illustration of Theorem~\ref{thm:low_density}. For small background density $\rho$, a minimizer $\Omega$ will consist of many droplets of size one and relative distances of order $\rho^{-1/3}$. \label{fig:big-box}}
\end{figure}

An important difficulty for the proof of Theorem~\ref{thm:low_density} is that we have many conjectures but little rigorous information about $\mu_*$, $m_*$ and $e_{\rm Jel}$. Remarkably, our proof is independent of the validity of these conjectures. It would certainly be easier if they had been solved.

Let us mention some previous works related to Theorem~\ref{thm:low_density}. The first term $\mu_*\rho$ was derived in~\cite{KnuMurNov-16} in a very dilute regime where the limits $L\to\ii$ and $\rho\to0$ are coupled in the manner $\rho\sim 1/L^2$. The leading order in the regime considered here (first $L\to\ii$ and then $\rho\to0$) was derived in~\cite{EmmFraKon-20}. The ultra-dilute regime with coupled limits $\rho\sim 1/L^3$ was considered earlier in~\cite{ChoPel-10} for the periodic model and a two-term expansion in the same spirit as~\eqref{eq:low_density_expansion} was proved. This is a regime where screening effects are much less pronounced than in the situation considered in this paper.

Our proof of~\eqref{eq:low_density_expansion} goes by showing upper and lower bounds on $e(\rho)$. We can replace the $o(\rho^{4/3})$ by an explicit error bound in the upper bound, see~\eqref{eq:final_upper_bound} below, but our proof of the lower bound does not provide any explicit bound on the error. It is an open problem to derive one.

The main argument we use for proving~\eqref{eq:low_density_expansion} is to reduce to smaller sets of size $A\rho^{-1/3}$ with $A$ large but fixed. In any such set we expect finitely many droplets, similar to the regime studied in~\cite{ChoPel-10}. For the upper bound we involve the periodic Jellium problem defined as in~\eqref{eq:thermo_limit_periodic} but for point-like particles. For the lower bound we use the Graf--Schenker inequality as in~\cite{GraSch-95,HaiLewSol_thermo-09} to remove the interactions between the different sets. We then argue by compactness to expand the energy locally in the corresponding sets of size $A\rho^{-1/3}$.

Several phase transitions are believed to occur when the density is increased from $0$ to $1/2$, that is, when going deeper into the crust of the neutron star. Those are often called \emph{``pasta phases''}~\cite{BohWhe-39,RavPetWil-83,HasSekYam-84,Oyamatsu-93}. At small but positive $\rho$ the gnocchis are not believed to be perfect balls as they become in the limit $\rho\to0$. They will be slightly elongated with an asymmetry increasing with $\rho$. When the density is increased beyond a certain value, instead of gnocchis one expects infinitely long cylinders called ``spaghetti'', which are arranged along a two-dimensional lattice. (Just like the gnocchis may not be perfect balls, the base of the spaghetti cylinders may not be perfect disks.) When the density is further increased beyond another critical value (strictly less than $1/2$), instead of spaghettis one expects infinite slabs called ``lasagna'', which are arranged along a one-dimensional lattice. In particular, at $\rho=1/2$ nucleons are expected to occupy slabs that are separated by completely unoccupied slabs of the same thickness. When $\rho$ is further increased, the story goes in the reverse way for the complement $\R^d\setminus \Omega$ until we reach an anti-gnocchi phase with $\Omega$ occupying most of $\R^3$, with holes looking likes balls at distances $(1-\rho)^{-1/3}$ and an energy behaving like for $\rho\to0$, by symmetry. The energies in the spaghetti and lasagna regimes are given by the one- and two-dimensional analogues of the model studied here. The one-dimensional analogue of $e(\rho)$ can be computed explicitly~\cite{RenWei-00,FraLie-19}. The two-dimen\-sio\-nal problem is still open and for partial results we refer to~\cite{CheOsh-07,GolMurSer-13,GolMurSer-14,Emmert-25_PhD}

In the remaining two subsections we define the constants $\mu_*$, $m_*$, $e_{\rm Jel}$ appearing in the statement of Theorem~\ref{thm:low_density}.

\subsection{Isolated droplets}
When we take $\rho\to0$ and replace $\Lambda$ by $\R^3$, we obtain the liquid drop model for a single droplet of nuclear matter without any electronic background. This leads us to introducing
\begin{equation}
\boxed{\mu_*:=\min_{\substack{\Omega\subset\R^3\\ |\Omega|>0}}I[\Omega],\qquad I[\Omega]:=|\Omega|^{-1}\left(\Per(\Omega)+\frac12\iint_{\Omega^2}\frac{\dx\,\dy}{|x-y|}\right).}
\label{eq:def_mu_*}
\end{equation}
The constant $\mu_*$ is the one appearing in the leading term of the low-density expansion~\eqref{eq:low_density_expansion}.

It is known~\cite{KnuMur-13,FraLie-15,KnuMurNov-16} that the minimum in~\eqref{eq:def_mu_*} is attained for some bounded set $\Omega_*$, and that the problem admits no minimizer of a too large volume or too large diameter (but the infimum can be the same because we can put several copies of $\Omega_*$ infinitely far apart). This allows us to introduce the largest possible volume of minimizers:
\begin{equation}
m_*:=\max\left\{|\Omega|\ :\ I[\Omega]=\mu_*\right\}.
\label{eq:def_m_*}
\end{equation}
It is also known that $\{|\Omega|\ :\ I[\Omega]=\mu_*\}$ is a closed bounded set in $(0,\ii)$, so that there is no minimizer of small volume either. For later purposes we introduce the smallest mass of minimizers:
\begin{equation}
m_{**}:=\min\left\{|\Omega|\ :\ I[\Omega]=\mu_*\right\}.
\label{eq:def_m_**}
\end{equation}

It is conjectured that $\Omega_*$ is unique up to symmetries, hence is a ball.
After optimization over the radius this leads to the conjecture that
\begin{equation}
\mu_*\overset{\text{\bf ?}}=9\left(\frac\pi{15}\right)^{\frac13},\qquad m_{**}\overset{\text{\bf ?}}=m_*\overset{\text{\bf ?}}=\frac52.
\label{eq:conjecture_m_*}
\end{equation}
It is rigorously known that $m_*\leq 8$ by~\cite{FraKilNam-16} and $m_{**}\geq 5/2$ by~\cite[Prop.~4.1]{FraLie-15}. As an aside we mention that minimizers of $I[\Omega]$ are known to be balls if we add the constraint that $|\Omega|=m$, with given $m\leq1$~\cite{ChoRuo-25}.

For our proof of Theorem~\ref{thm:low_density}, we will need to know the rather precise behavior of minimizing sequences for the minimization problem $\mu_*$ in~\eqref{eq:def_mu_*}. The following result is an improvement of~\cite{FraLie-15,KnuMurNov-16}.

\begin{theorem}[Decomposition of minimizing sequences]\label{thm:decomposition}
Let $\Omega_n \subset\R^3$ be a sequence of Borel sets such that $|\Omega_n|\to m>0$ and
$$\lim_{n\to\ii}\frac{I[\Omega_n]}{|\Omega_n|}=\mu_*.$$
Then there exist an integer $K$ and Borel sets $\Omega^{(1)}_*,...,\Omega^{(K)}_*\subset\R^3$ with finite positive measure $|\Omega^{(j)}_*|=m^{(j)}_*>0$, sequences $a_n^{(j)}\in\R^3$ with $|a_n^{(j)}-a_n^{(k)}|\to_{n\to\ii}\ii$ for $j\neq k$ and a constant $R>0$ such that the following hold, after extraction of a subsequence:
\begin{itemize}
\item[$(i)$] each $\Omega^{(j)}_*$ is a minimizer for $\mu_*$, that is $I[\Omega^{(j)}_*]=\mu_*|\Omega^{(j)}_*|$, and in particular $m_*^{(j)}\in[m_{**},m_*]$;
\item[$(ii)$] for every $j\in\{1,...,K\}$, the set $(\Omega_n-a_n^{(j)})\cap B(0,R)$ converges strongly to $\Omega_*^{(j)}$, in the sense that the corresponding characteristic functions converge in $L^1$;
\item[$(iii)$] the volume and perimeter of the remaining part are negligible:
\begin{equation}
\left|\Omega_n\setminus \bigcup_{j=1}^KB(a_n^{(j)},R)\right|\to0,\qquad \Per\bigg(\Omega_n\setminus \bigcup_{j=1}^KB(a_n^{(j)},R)\bigg)\to0
\label{eq:outside_volume}
\end{equation}
and hence $m=\sum_{j=1}^K m_*^{(j)}$;
\item[$(iv)$] if $K\geq2$ we have
\begin{equation}
I[\Omega_n]\geq \mu_*|\Omega_n|+\sum_{1\leq j<k\leq K}\frac{m_*^{(j)}m_*^{(k)}+o(1)}{|a_n^{(j)}-a_n^{(k)}|}.
\label{eq:repulsion}
\end{equation}
\end{itemize}
\end{theorem}

The theorem gives the behavior of an arbitrary minimizing sequence for the liquid drop problem~\eqref{eq:def_mu_*}. The statement is that an almost optimal set $\Omega_n$ has to look like the union of finitely many minimizers placed at a large distance from each other, up to a negligible set (in terms of volume and perimeter). In the case that there are several pieces, the theorem gives in addition a lower bound on the interaction between them, in terms of the relative distances.

We first give an easy consequence of this theorem in the case that the minimizing sequence $\Omega_n$ is known to have diameter $\ell_n\to\ii$ and an energy that is correct up to a $O(1/\ell_n)$. This is the version that will be used later in our proof of the lower bound in Theorem~\ref{thm:low_density}.

\begin{corollary}\label{cor:repulsion_ell_n}
Let $A\subset\R^3$ be a fixed bounded open set. Let $\Omega_n \subset\R^3$ be a sequence of Borel sets such that $\Omega_n\subset \ell_n A$, $|\Omega_n|\to m>0$ and
$$I[\Omega_n]\leq |\Omega_n|\mu_*+\frac{C}{\ell_n},$$
for some $\ell_n\to\ii$. Then, after extraction of a further subsequence, the centers from Theorem~\ref{thm:decomposition} satisfy $\ell_n^{-1}a_n^{(j)}\to y^{(j)}$ for $K$ distinct points $y^{(1)},...,y^{(K)}$ in $\overline{A}$. When $K\geq2$ we have
\begin{equation}
I[\Omega_n]\geq \mu_*|\Omega_n|+\frac1{\ell_n}\sum_{1\leq j<k\leq K}\frac{m_*^{(j)}m_*^{(k)}}{|y^{(j)}-y^{(k)}|}+o\left(\frac1{\ell_n}\right).
\label{eq:repulsion2}
\end{equation}
\end{corollary}

The corollary says that if there are several droplets they have to be as far as possible, that is, the distances are all proportional to the size of the piece under consideration. This is of course due to the Coulomb repulsion.

Since $m_*^{(j)}\geq m_{**}$, it follows from this corollary that if $C$ is small enough, then there can be only one piece ($K=1$).

\begin{proof}[Proof of Corollary~\ref{cor:repulsion_ell_n}]
We apply Theorem~\ref{thm:decomposition} to the sequence $\Omega_n$. Since $\Omega_n$ is included in $\ell_n A$, after extraction of a subsequence we can assume $\ell_n^{-1}a_n^{(j)}\to y^{(j)}\in \overline{A}$. If $K\geq2$, the lower bound~\eqref{eq:repulsion} and the assumed upper bound on $I[\Omega_n]$ imply
$$\frac{|a_n^{(j)}-a_n^{(k)}|}{\ell_n}\geq \frac{C}{2m_*^{(j)}m_*^{(k)}}$$
and hence the $y^{(j)}$'s are distinct.
\end{proof}

\subsection{Jellium model}
We conclude this section by presenting the \emph{Jellium model} describing classical charged particles moving in a uniform background. This is the model we will see when ``zooming out'' at scale $\rho^{-1/3}$, where the density of electrons is equal to 1 and the nuclear droplets are very peaked and behave like Dirac delta's.

Jellium is often associated with the name of Wigner~\cite{Wigner-34} who however introduced the quantum version in 1934. The classical model is also called the \emph{one-component plasma} and it seems to have been proposed around 1900 by J.~J.~Thomson~\cite{Thomson-04} -- based on previous ideas of W.~Thomson (Lord Kelvin) -- in order to describe the electrons in an atom, before the discovery of the nuclei. In this context, it is sometimes called the \emph{plum pudding model}.

In the thermodynamic limit, the energy per unit volume is defined similarly as~\eqref{eq:thermo_limit} by
\begin{multline}
e_\text{Jel}:=\lim_{N\to\ii}N^{-1}\min_{x_1,...,x_N\in N^{\frac13}\Lambda_1}\bigg\{\sum_{1\leq j<k\leq N}\frac{1}{|x_j-x_k|}\\
-\sum_{j=1}^N\int_{N^{\frac13}\Lambda_1}\frac{\dy}{|x_j-y|}+\frac12\iint_{(N^{\frac13}\Lambda_1)^2}\frac{\dx\,\dy}{|x-y|}\bigg\},
\label{eq:Jellium}
\end{multline}
with $\Lambda_1$ any fixed smooth enough open set of volume $|\Lambda_1|=1$, for instance a ball or a cube. The $x_j$ are the locations of the classical particles.

The existence of the limit~\eqref{eq:Jellium} was proved first by Lieb and Narnhofer in 1975~\cite{LieNar-75}, based on a method introduced earlier by Lieb and Lebowitz in~\cite{LieLeb-72}. This model has recently reappeared in many works mentioned in the review paper~\cite{Lewin-22}. In particular, in a long series of works culminating in~\cite{Serfaty-24_lecture_notes} Serfaty and co-workers have obtained detailed results about the Jellium model with variable background density. As a byproduct, also results about the original Jellium model were obtained.

It is known that
\begin{equation}
-1.4508\approx-\frac35\left(\frac{9\pi}{2}\right)^{\frac13}\leq e_{\text{Jel}}\leq \zeta_\text{BCC}(1)\simeq-1.4442
\label{eq:bounds_e_Jellium}
\end{equation}
and conjectured that the upper bound is in fact an equality~\cite{Wigner-34,BlaLew-15,Lewin-22}. The lower bound is proved in~\cite[Appendix]{LieNar-75}. The right side of~\eqref{eq:bounds_e_Jellium} is the Epstein zeta function of the Body-Centered Cubic lattice~$\cL$ of unit density, containing the origin, given by the expression
$$\zeta_{\cL}(s)=\frac12\sum_{\ell\in\cL\setminus\{0\}}\frac1{|\ell|^s}.$$
This is a convergent series for $\Re(s)>3$ which admits a meromorphic extension to the whole complex plane $\C$ with a unique pole at $s=3$. The Jellium value is then given by this extension at $s=1$~\cite[Lemma~32]{Lewin-22}.

\subsection*{Organisation of the paper}
The rest of the paper is devoted to the proofs of Theorems~\ref{thm:thermo_limit},~\ref{thm:low_density} and~\ref{thm:decomposition}. We start in Section~\ref{sec:proof_decomposition} with the proof of Theorem~\ref{thm:decomposition}. Sections~\ref{sec:proof_upper_bound} and~\ref{sec:proof_lower_bound} are devoted to the upper and lower bound in Theorem~\ref{thm:low_density}, respectively. In Appendix~\ref{app:proof_thermo-limit} we describe how one can prove the existence of the thermodynamic limit in Theorem~\ref{thm:thermo_limit} following ideas of Lieb, Lebowitz and Narnhofer. Throughout the proofs we use the notation
$$D(f,g):=\frac12\iint_{\R^6}\frac{f(x)g(y)}{|x-y|}\dx\,\dy,\qquad D(f):=D(f,f)$$
for the Coulomb energy.

\subsection*{Acknowledgment}
We thank Christian Hainzl for useful discussions regarding the grand-canonical liquid drop model and the use of the Graf--Schenker inequality. Partial support through US National Science Foundation grant DMS-1954995 (R.L.F.), as well as through the German Research Foundation through EXC-2111-390814868 (R.L.F.) and TRR 352-Project-ID
470903074 (R.L.F.~and R.S.) is acknowledged. Part of this work was completed when R.L.F.~was visiting Ceremade on a Chaire d'Excellence of the Fondation Sciences Mathématiques de Paris (FSMP).

\section{Proof of Theorem~\ref{thm:decomposition} on minimizing sequences for the liquid drop model}\label{sec:proof_decomposition}
In this section we provide the proof of Theorem~\ref{thm:decomposition}. The idea is to iterate several times the proof of~\cite{FraLie-15} so as to optimize the error terms.

\subsection*{Step 1. Extracting bubbles}
The volume $|\Omega_n|$ stays bounded and does not tend to 0 by assumption. Since the energy $I[\Omega_n]$ converges we conclude that $\Per(\Omega_n)$ is bounded. By~\cite[Prop.~2.1]{FraLie-15} we can find a center $a_n^{(1)}$ so that (after extraction of a subsequence) $\Omega_n-a_n^{(1)}\to \Omega_*^{(1)}\neq\emptyset$ locally with $0<m_*^{(1)}:=|\Omega_*^{(1)}|\leq m$ and $\Per(\Omega_*^{(1)})\leq \liminf_{n\to\ii}\Per(\Omega_n)$. Here local convergence of sets means $L^1_{\rm loc}$ convergence of the corresponding characteristic functions.

Next we consider the Levy function $M_n(r):=|\Omega_n\cap B(a_n^{(1)},r)|$ which is non-decreasing in $r$, uniformly bounded, and converges for every fixed $r$ to $M_\ii(r):=|\Omega_*^{(1)}\cap B(0,r)|$. We can thus find a sequence $R_n\to\ii$ such that
$$|\Omega_n\cap B(a_n^{(1)},R_n)| \to m_*^{(1)},\qquad \big|\Omega_n\cap \{R_n\leq |x-a_n^{(1)}|<2R_n\}\big|=:\eps_n \to0. $$
The latter implies that we can find a radius $R_n'\in[R_n,2R_n)$ such that
\begin{equation}
\cH^2\big(\Omega_n\cap \partial B(a_n^{(1)},R_n')\big)\leq \frac{\eps_n}{R_n},
\label{eq:measure_boundary}
\end{equation}
where $\cH^2$ denotes the two-dimensional Hausdorff measure. Writing $\Omega_n=\Omega_n^i\cup \Omega_n^o$ (disjoint union) with $\Omega_n^i:=\Omega_n\cap B(a_n^{(1)},R_n')$ we obtain
\begin{align}
&I[\Omega_n]-\mu_*|\Omega_n|\nn\\
&\qquad =I[\Omega_n^i]-\mu_*|\Omega_n^i|+I[\Omega_n^o]-\mu_*|\Omega_n^o|-2\cH^2\big(\Omega_n\cap \partial B(a_n^{(1)},R_n')\big)\nn\\
&\qquad\qquad+2D(\1_{\Omega_n^i},\1_{\Omega_n^o})\nn\\
&\qquad \geq I[\Omega_n^i]-\mu_*|\Omega_n^i|+I[\Omega_n^o]-\mu_*|\Omega_n^o|-2\frac{\eps_n}{R_n}+2D(\1_{\Omega_n^i},\1_{\Omega_n^o})\nn\\
&\qquad \geq I[\Omega_n^i]-\mu_*|\Omega_n^i|-2\frac{\eps_n}{R_n}.\label{eq:repulsion_pre}
\end{align}
To get the final lower bound we used that
$$I[\Omega_n^o]-\mu_*|\Omega_n^o|\geq0,\qquad D(\1_{\Omega_n^i},\1_{\Omega_n^o})\geq0.$$
There is a small subtlety in the first equality in~\eqref{eq:repulsion_pre}. From~\cite[Lem.~15.12]{Maggi-12}, the equality holds almost everywhere in $R_n'$ and therefore we need to use that~\eqref{eq:measure_boundary} in fact holds for a positive measure set of $R_n'\in[R_n,2R_n)$.

Since the left side of~\eqref{eq:repulsion_pre} tends to 0 by assumption, we can use the lower semi-continuity of $I$ proved in~\cite{FraLie-15} and we arrive at
$$I[\Omega_*^{(1)}]\leq \mu_*|\Omega_*^{(1)}|.$$
Recall that $|\Omega_n^i|\to |\Omega_*^{(1)}|>0$ by construction of $R_n$.
This shows that $\Omega_*^{(1)}$ is a minimizer. In particular we have $m_*^{(1)}\geq m_{**}>0$ since those have a minimal mass. Going back the the above estimates we see that
$$I[\Omega_n^o]-\mu_*|\Omega_n^o|\to 0.$$
Hence we have two possibilities. If the volume $|\Omega_n^o|\to0$ we can stop and let $K:=1$. The fact that $I[\Omega_n^o]\to0$ implies in addition that $\Per(\Omega_n^o)\to0$. If $|\Omega_n^o|\nrightarrow 0$, we apply all the previous arguments to $\Omega_n^o$ and find a new set $\Omega_*^{(2)}$. We proceed this way by induction. We know that the procedure will terminate after a finite number of steps $K\leq m/m_{**}$. We thus obtain our $K$ minimizers $\Omega_*^{(j)}$.

\subsection*{Step 2. Cutting at a finite distance.} So far we have found radii $R_n^{(j)}\to\ii$ and centers $a_n^{(j)}$ so that
$$\big(\Omega_n-a_n^{(j)}\big)\cap B(0,R_n^{(j)})\to\Omega_*^{(j)}\quad\text{globally}$$
in the sense that the characteristic functions converge in $L^1(\R^3)$ and
$$\left|\Omega_n\setminus \bigcup_{j=1}^K B(a_n^{(j)},R_n^{(j)})\right|\to0.$$
In the statement $(ii)$ of the theorem we had a fixed cutoff radius $R$ instead of divergent sequences. Let us quickly explain how to do this.

We know that minimizers for $\mu_*$ have a finite diameter~\cite[Lemma~7.2]{KnuMur-13}. We can thus find an $R>0$ such that $\Omega_*^{(j)}\subset B(0,R)$ for all $j=1,...,K$. From the $L^1$ convergence of the characteristic functions this implies
$$\left|(\Omega_n-a_n^{(j)}\big)\cap \{R\leq|x|\leq R_n^{(j)}\}\right|\to0$$
and therefore we obtain as we wanted
\begin{equation}
\eps'_n:=\left|\Omega_n\setminus \bigcup_{j=1}^KB(a_n^{(j)},R)\right|\to0.
\label{eq:remains_small_vol}
\end{equation}
Arguing exactly as in Step 1 for the energy, we can cut at a common distance $R\leq R'\leq 2R$ around each point $a_n^{(j)}$ and find again $I[\Omega_n^o]\to0$, which implies $\Per(\Omega_n^o)\to0$. We have thus proved~\eqref{eq:outside_volume}. This concludes the proof of Theorem~\ref{thm:decomposition} when $K=1$.

\subsection*{Step 3. Computation of the interaction for $K\geq2$.}
Next we modify our localization procedure once more to evaluate more precisely the interactions between the $\Omega_*^{(j)}$'s in the case $K\geq2$.

For each $j$ we call $\delta_n^{(j)}:=\min_{k\neq j}|a_n^{(k)}-a_n^{(j)}|\to\ii$ the distance to the nearest bubble and we localize with respect to this distance. By the same argument as in the previous step, we can find new radii $\delta_n^{(j)}/6\leq r_n^{(j)}\leq \delta_n^{(j)}/3$ such that
$$\cH^2\big(\Omega_n\cap\partial B(a_n^{(j)},r_n^{(j)})\big)\leq\frac{6\eps_n'}{\delta_n^{(j)}}. $$
We then define the new sets
$$\Omega_n^{(j)}:=\Omega_n\cap B(a_n^{(j)},r_n^{(j)}),\qquad \Omega_n^o:=\Omega_n\setminus \bigcup_{j=1}^K\Omega_n^{(j)}.$$
We decompose the energy as above and obtain
\begin{align}
I[\Omega_n]-\mu_*|\Omega_n|\nn
&\geq I[\Omega_n^o]-\mu_*|\Omega_n^o| -6\eps'_n\sum_{j=1}^K\frac{1}{\delta_n^{(j)}} +2\sum_{j\neq k} D(\1_{\Omega_n^{(j)}},\1_{\Omega_n^{(k)}})\nn\\
&\geq \sum_{j\neq k} \frac{m_*^{(j)}m_*^{(k)}+o(1)}{|a_n^{(j)}-a_n^{(k)}|}.\label{eq:repulsion_end_of_proof}
\end{align}
In the last line we have used that $I[\Omega_n^o]\geq \mu_*|\Omega_n^o|$ and we bounded the interaction from below by
\begin{align*}
2D(\1_{\Omega_n^{(j)}},\1_{\Omega_n^{(k)}})&\geq \int_{\Omega_n\cap B(a_n^{(j)},R)}\int_{\Omega_n\cap B(a_n^{(k)},R)}\frac{\dx\,\dy}{|x-y|}\\
&\geq \frac{|\Omega_n\cap B(a_n^{(j)},R)||\Omega_n\cap B(a_n^{(k)},R)|}{|a_n^{(j)}-a_n^{(k)}|+2R}\\
&= \frac{m_*^{(j)}m_*^{(k)}+o(1)}{|a_n^{(j)}-a_n^{(k)}|}.
\end{align*}
We also used that the localization error can be bounded by
\begin{equation}
\frac{6\eps_n}{\delta_n^{(j)}}\leq6\eps_n\sum_{k\neq j}\frac{1}{|a_n^{(k)}-a_n^{(j)}|}.
\label{eq:control_loc_error}
\end{equation}
The bound in~\eqref{eq:repulsion_end_of_proof} is our claimed estimate.
This concludes the proof of Theorem~\ref{thm:decomposition}.
\qed

\begin{remark}[Riesz case]
\rm All the arguments in the proof of Theorem~\ref{thm:decomposition} apply if the Coulomb potential $1/|x|$ is replaced by a Riesz potential $1/|x|^s$ with $0<s<d$ in any dimension $d\geq1$, except for~\eqref{eq:repulsion}. The latter requires $0<s\leq 1$ so that the localization error of the perimeter can be controlled as in~\eqref{eq:control_loc_error}. Moreover, the finiteness of $m_*$ is only known for $s\leq2$~\cite{FraNam-21}.
\end{remark}

\section{Proof of the upper bound in Theorem~\ref{thm:low_density}}\label{sec:proof_upper_bound}
For $N$ a large but fixed integer, we consider the box $C_\ell=(-\ell/2,\ell/2)^3$ of volume $|C_\ell|=m_*N/ \rho$ where we recall that $m_*$ was defined in~\eqref{eq:def_m_*}. In other words we take
$$\ell=\left(\frac{m_*N}\rho\right)^{1/3}.$$
In this box, we place $N$ points $X=\{x_1,...,x_N\}$ optimizing the \emph{periodic Jellium problem}
$$E_{\rm per}(\ell):=\min_{x_1,...,x_N\in C_\ell}\sum_{1\leq j<k\leq N} G_\ell(x_j-x_k).$$
Recall that $G_\ell$ is $\ell\Z^3$ periodic with the Fourier coefficients $\widehat{G_\ell}(k)=\sqrt{2/\pi}|k|^{-2}$ for $k\in(2\pi/\ell)\Z^3$ and $\widehat{G_\ell}(0)=0$. We have $G_\ell(x)=G_1(x/\ell)/\ell$. Due to the periodicity of $G_\ell$ we know that $X+\tau$ modulo $\ell\Z^3$ is a minimizer for all $\tau\in\R^3$. Hence we can without loss of generality assume that
$$\sum_{j=1}^Nx_j=0.$$
An argument of Lieb shows that $|x_j-x_k|\geq c(m_*/\rho)^{1/3}$ for some constant $c>0$ independent of $N$ and all $j\neq k$ (see~\cite[Lemma~25]{Lewin-22}). It will be useful to consider in addition a different configuration of points $Y=\{y_1,...,y_N\}\subset C_\ell$ which also satisfies $\sum_{j=1}^Ny_j=0$ and the additional condition that $\rd(y_j,\partial C_\ell)\geq c(m_*/\rho)^{1/3}$, that is, staying far away from the boundary. For concreteness we can place them on a lattice of side length $\eps(m_*/\rho)^{1/3}$ for some $\eps>0$ small enough. We can do this in such a way that $|y_j-y_k|\geq c(m_*/\rho)^{1/3}$ for all $j\neq k$.

We consider the points
$$Z:=\bigcup_{\substack{\xi\in\Z^3\\|\xi|_\ii\leq k}}(X+\ell \xi)\cup \bigcup_{\substack{\xi'\in\Z^3\\|\xi'|_\ii= k+1}}(Y+\ell \xi').$$
In words we place $(2k+1)^3$ copies of $X$ in a periodic fashion over a large cube of side length $(2k+1)\ell$. We also add a layer of cubes around the big cube where we instead use the $y_j$'s. The full system is contained in the larger cube
$$\Lambda:=C_{(2k+3)\ell}$$
of side length $(2k+3)\ell$ and the total number of points is $K:=\#Z=(2k+3)^3N$. By construction the points of $Z$ are at distance $\geq c(m_*/\rho)^{1/3}$ from each other and from the boundary of $\Lambda$. We used the external layer to ensure that the points are all strictly inside $\Lambda$.

Now we give ourselves an optimal set $\Omega_*\subset B(0,r_*)$ for $\mu_*$ in~\eqref{eq:def_mu_*}, of volume $|\Omega_*|=m_*$, and consider the trial set
$$\Omega:=\bigcup_{z\in Z}(z+R_z\Omega_*).$$
That is, at each site $z\in Z$ we place a rotated copy of $\Omega_*$ with a rotation $R_z\in SO(3)$. For $\rho$ small enough (that is, $\ell$ large enough), the copies of $\Omega_*$ do not intersect and are completely included in $\Lambda$, uniformly for all rotations $R_z\in SO(3)$. We obtain the upper bound
\begin{align*}
E_\Lambda(\rho)&\leq K\Per(\Omega_*)+D(\1_\Omega-\rho\1_{\Lambda})\\
&= K\mu_* m_*+\sum_{z\neq z'\in Z}D(\1_{z+R_z\Omega_*},\1_{z'+R_{z'}\Omega_*})\\
&\qquad -2\rho\sum_{z\in Z}D(\1_{z+R_z\Omega_*},\1_{\Lambda})+\rho^2D(\1_\Lambda,\1_\Lambda).
\end{align*}
Next we average this bound over the independent rotations $R_z$'s. Recall that Newton's theorem~\cite[Thm. 9.7]{LieLos-01} implies
\begin{equation}
\fint_{SO(3)}\int_{\R^d} \frac{\1_{z+R\Omega_*}(y)}{|x-y|}\,\dy\,\rd R=\frac{m_*}{|x-z|}\qquad\text{for $|x-z|\geq r_*$.}
\label{eq:Newton_averaged}
\end{equation}
Hence
$$\fint_{SO(3)}\fint_{SO(3)} D(\1_{z+R\Omega_*},\1_{z'+R'\Omega_*})\,\rd R\,\rd R'=\frac{m_*^2}{2|z-z'|}$$
and
\begin{align*}
&2\fint_{SO(3)} D(\1_{z+R\Omega_*},\1_\Lambda)\,\rd R\\
&\qquad\qquad =m_*\int_\Lambda \frac{\dy}{|z-y|}+\fint_{SO(3)} \int_{R\Omega_*}\int_{B(0,r_*)}\left(\frac{1}{|x-y|}-\frac1{|y|}\right)\dy\,\dx\,\rd R\\
&\qquad\qquad \geq m_*\int_\Lambda \frac{\dy}{|z-y|}-m_*\int_{B(0,r_*)}\frac{\dy}{|y|}\\
&\qquad\qquad = m_*\int_\Lambda \frac{\dy}{|z-y|}-2\pi r_*^2 m_*.
\end{align*}
As a conclusion we have
\begin{align*}
E_\Lambda(\rho)&\leq  K\mu_* m_*+\frac12\sum_{z\neq z'\in Z}\frac{m_*^2}{|z-z'|}\\
&\qquad -\rho m_*\sum_{z\in Z}\int_\Lambda\frac{\dy}{|z-y|}+\frac{\rho^2}2\iint_{\Lambda^2}\frac{\dx\,\dy}{|x-y|}+2\pi r_*^2 m_*\rho K.
\end{align*}
Next we estimate the terms associated with the cubes located on the outermost layer. For $\xi\in\Z^3$, we denote
$$\mu_\xi:=\begin{cases}
m_*\sum_{x\in X}\delta_{\ell \xi+x}-\rho\1_{\ell \xi+C_\ell}&\text{if $|\xi|_\ii\leq k$}\\
m_*\sum_{y\in Y}\delta_{\ell \xi+y}-\rho\1_{\ell \xi+C_\ell}&\text{if $|\xi|_\ii= k+1$}.
         \end{cases}
$$
Since $\int \rd\mu_\xi=0$ and $\int x\,\rd\mu_\xi(x)=0$ for all $\xi$, we have
\begin{equation}
|D(\mu_{\xi_0},\mu_{\xi})|\leq \frac{C_N}{\ell |\xi_0-\xi|^5},\qquad \forall \xi\neq \xi_0,\qquad|\xi_0|_\ii=k+1.
\label{eq:quadrupole_interaction}
\end{equation}
Indeed, if $\xi$ is not a neighbor of $\xi_0$ this is~\cite[Lemma~2.1]{AnaLew-20} and the constant $C_N$ can be bounded by $CN^2$. If the two cubes touch, we need to use that $Y$ is at distance $(m_*/\rho)^{1/3}=\ell N^{-1/3}$ to the boundary and we obtain the bound
$$\left|\sum_{j=1}^N\int \frac{\rd\mu_\xi(x)}{|y_j-x|}\right|\leq C(\rho/m_*)^{1/3}N^2=\frac{CN^{7/3}}{\ell}.$$
Similarly, the interaction with the background is
$$\left|\rho\int_{\ell\xi_0+C_\ell} \dy\int\frac{\rd\mu_\xi(x)}{|y-x|}\right|\leq C\rho\ell^2N=\frac{Cm_*N^2}{\ell}$$
because we have the uniform bound $\1_{C_\ell}\ast{|x|}^{-1}\leq C\ell^2$. The latter follows for instance from Newton's theorem~\cite[Thm. 9.7]{LieLos-01}, using
\begin{equation}
\int_{C_\ell}\frac{\dy}{|x-y|}\leq \int_{B(0,r\ell)}\frac{\dy}{|x-y|}=\int_{B(0,r\ell)}\frac{\dy}{\max(|x|,|y|)}\leq \int_{B(0,r\ell)}\frac{\dy}{|y|}=2\pi r^2\ell^2
\label{eq:Newton}
\end{equation}
where $r={\rm diam}(C_1)/2$. This is how we can get~\eqref{eq:quadrupole_interaction} with a constant $C_N$ depending only on $N$.

The series being convergent due to~\eqref{eq:quadrupole_interaction}, we conclude that
$$\sum_{\xi\neq \xi_0}|D(\mu_{\xi_0},\mu_{\xi})|\leq \frac{C_N'}{\ell}$$
for any fixed $\xi_0$ with $|\xi_0|_\ii=k+1$. Hence, denoting by
$$\tilde Z:=\bigcup_{\substack{\xi\in\Z^3\\|\xi|_\ii\leq k}}(X+\ell \xi),\qquad\tilde\Lambda=C_{(2k+1)\ell}$$
the points inside and the background with the outer layer removed, we obtain
\begin{multline*}
E_\Lambda(\rho)\leq  K\mu_* m_*+\frac12\sum_{z\neq z'\in \tilde Z}\frac{m_*^2}{|z-z'|}-\rho m_*\sum_{z\in \tilde Z}\int_{\tilde\Lambda}\frac{\dy}{|z-y|}\\
+\frac{\rho^2}2\iint_{\tilde\Lambda^2}\frac{\dx\,\dy}{|x-y|}+2\pi r_*^2 m_*\rho K+\big((2k+3)^3-(2k+1)^3\big)\left(\frac{C'_N}{\ell}+\eta_Y\right)
\end{multline*}
where
$$\eta_Y=\sum_{1\leq j<k\leq N}\frac1{|y_j-y_k|}-\sum_{j=1}^N\int_{C_\ell}\frac{\dy}{|y_j-y|}+\frac12\iint_{(C_\ell)^2}\frac{\dx\,\dy}{|x-y|}$$
is the energy of one cube from the outer layer, and $(2k+3)^3-(2k+1)^3$ equals the number of such cubes.
Next we divide by the volume $|\Lambda|=(2k+3)^3\ell^3=(2k+3)^3m_*N/\rho$ and we take the limit $k\to\ii$ with all the other parameters fixed. For the Jellium term we use the convergence to the periodic energy as stated for instance in~\cite[Lemma 32]{Lewin-22}. We obtain
\begin{equation}
e(\rho)\leq  \mu_* \rho+\frac{m_*^2}{\ell^3}\left(\sum_{1\leq j<k\leq N}G_\ell(x_j-x_k)+\frac{M_{\Z^3}}{2\ell}\right)+2\pi r_*^2\rho^2
\end{equation}
where $M_{\Z^3}=\lim_{x\to0}(G_1(x)-|x|^{-1})$ is the Madelung constant of the lattice $\Z^3$. As a last step we take the limit $\ell\to\ii$, using that the periodic Jellium problem converges to $e_{\rm Jel}$ in~\eqref{eq:Jellium}, by~\cite[Section V]{LewLieSei-19b} and~\cite[Theorem~40]{Lewin-22}. There is an additional factor $(\rho/m_*)^{4/3}$ since the effective density of the periodic problem is $\rho/m_*$. Our final bound is
\begin{equation}
\boxed{e(\rho)\leq  \mu_*\, \rho+m_*^{\frac23}\,e_{\rm Jel}\,\rho^{\frac43}+2\pi r_*^2\rho^2.}
\label{eq:final_upper_bound}
\end{equation}

\section{Proof of the lower bound in Theorem~\ref{thm:low_density}}\label{sec:proof_lower_bound}

Our proof of the lower bound relies on the grand-canonical version of both the liquid drop model and the Jellium problem. The former consist of minimizing over \textbf{all sets $\Omega$ without any constraint on its volume}, after we have removed the expected leading term $\mu_*|\Omega|$. The corresponding minimal energy thus reads
\begin{equation}
\boxed{F_\Lambda(\rho):=\min_{\Omega\subset\Lambda}\Big\{\Per(\Omega)+D(\1_\Omega-\rho\1_\Lambda)-\mu_*\,|\Omega|\Big\}.}
\label{eq:def_F}
\end{equation}
The name ``grand-canonical'' comes from statistical physics and $\mu_*$ is interpreted as a ``chemical potential''. For us it is an essential concept because the relaxation of the volume constraint will allow for the subsequent localization procedure (the volume constraint is global and not satisfied locally).

We have $E_\Lambda(\rho)\geq F_\Lambda(\rho)+\mu_*\rho|\Lambda|$ and the goal is to prove that, in the thermodynamic limit, $F_\Lambda(\rho)$ converges to the Jellium energy, after taking $\rho\to0$. We split the argument into three steps. First we reduce to sets of side length $A\rho^{-1/3}$ with $A$ a large but fixed constant, before we actually take the limit $\rho\to0$. Finally, we take $A\to\ii$, using that the grand-canonical Jellium problem has the same thermodynamic limit as the canonical one.

\subsection*{Step 1. Reducing to sets of finite size}
In this first step we localize our problem to sets of finite size. For this we use the Graf--Schenker inequality~\cite{GraSch-95,HaiLewSol_thermo-09}. Locally we get the grand-canonical problem introduced in~\eqref{eq:def_F} because although we start with a globally neutral system ($|\Omega|=\rho|\Lambda|$), neutrality does not hold locally. The precise statement is the following.

\begin{proposition}[Reduction to sets of finite size]\label{prop:Graf--Schenker}
Let $\Delta\subset\R^3$ be a tetrahedron with $|\Delta|=1$. Then we have for every $\ell>0$
\begin{equation}
e(\rho)-\mu_*\rho\geq \frac{F_{\ell\Delta}(\rho)}{\ell^3}-C\frac{\rho}{\ell}.
\label{eq:lower_simplex}
\end{equation}
with a constant $C$ depending only on $\Delta$, where we recall that $F_{\ell\Delta}(\rho)$ is the grand-canonical problem defined in~\eqref{eq:def_F}.
\end{proposition}

\begin{proof}
Let $G:= \R^3\rtimes SO(\R^3)$, which acts on tetrahedra by translations and rotations. Elements of $G$ are denoted by $g$ and integration with respect to Haar measure by $\rd g$. The normalization is chosen so that
$$\frac1{\ell^3}\int_G \1_{g\ell\Delta}(x) \,\rd g=1$$
for all $\ell>0$ and $x\in\R^3$ (the left side is constant since we integrate over translations). Multiplying by $\1_\Omega$ and integrating we deduce that for any Borel set $\Omega$,
\begin{equation}
|\Omega| = \int_G |\Omega\cap g\ell\Delta| \,\frac{\rd g}{\ell^3}.
\label{eq:GS_volume}
\end{equation}
We show in the next lemma that with $C=\Per(\Delta)$ we have
\begin{equation}
\Per(\Omega) = \int_G \Per(\Omega\cap g\ell\Delta)\,\frac{\rd g}{\ell^3} - C \frac{|\Omega|}{\ell} \,.
\label{eq:GS_perimeter}
\end{equation}
Next we take a large cube $\Lambda_L:=(-L/2,L/2)^3$ and any domain $\Omega\subset\Lambda_L$. The Graf--Schenker inequality from~\cite{GraSch-95} states that
	\begin{multline*}
		\frac12 \iint_{\Lambda_L\times\Lambda_L} \frac{(\1_\Omega(x)-\rho)\,(\1_\Omega(y)-\rho)}{|x-y|}\,\dx\,\dy \\
		 \geq \int_G \frac{\rd g}{\ell^3} \ \frac12 \iint_{(\Lambda_L\cap g\ell\Delta) \times(\Lambda_L\cap g\ell\Delta)} \frac{(\1_\Omega(x)-\rho)\,(\1_\Omega(y)-\rho)}{|x-y|}\dx\,\dy \,.
	\end{multline*}
For this inequality to hold it is essential to average both over rotations and translations.
For fixed $g$ we expand the integrand in the remaining double integral.
Together with~\eqref{eq:GS_volume} and~\eqref{eq:GS_perimeter}, we obtain for all $\Omega\subset\Lambda_L$
	\begin{align}
		\mathcal E_{\Lambda_L}[\rho,\Omega] & \geq \int_G  \cE_{\Lambda_L\cap g\ell\Delta}[\Omega\cap g\ell\Delta]\frac{\rd g}{\ell^3}  - C \frac{|\Omega|}{\ell} \nn\\
		& = \int_G \left( \cE_{\Lambda_L\cap g\ell\Delta}[\Omega\cap g\ell\Delta]-\mu_*|\Omega\cap g\ell\Delta|\right)\frac{\rd g}{\ell^3} +\left( \mu_* - \frac{C}\ell\right) |\Omega| \,.\label{eq:GS}
	\end{align}
When $g$ is so that $g\ell\Delta\subset \Lambda_L$ we have the lower bound
$$\cE_{\Lambda_L\cap g\ell\Delta}[\Omega\cap g\ell\Delta]-\mu_*|\Omega\cap g\ell\Delta|\geq F_{g\ell\Delta}(\rho)=F_{\ell\Delta}(\rho)$$
with the last equality due to the invariance of the problem under rotations and translations. Let us denote by
$$I_{L,\ell}:=\big\{g\in G\ :\ g\ell\Delta\subset \Lambda_L\big\},\qquad  B_{L,\ell}:=\big\{g\in G\ :\ g\ell\Delta\cap \partial\Lambda_L\neq\emptyset\big\}$$
the sets corresponding to the interior and boundary tetrahedra. For $g\in B_{L,\ell}$ the corresponding tetrahedron intersects the boundary of the large cube $\Lambda_L$. Let denote by $(\partial\Lambda_L)_{R}:=\partial\Lambda_L+B(0,R)$ the set of points at distance $\leq R$ to the boundary of the large cube. Then we have $g\ell\Delta\subset (\partial\Lambda_L)_{R\ell}$ for a large enough universal constant $R$ and thus
\begin{align*}
\int_{B_{L,\ell}}\rd g=\int_{B_{L,\ell}}|g\ell\Delta|\frac{\rd g}{\ell^3}&=\int_{B_{L,\ell}}|(\partial\Lambda_L)_{R\ell}\cap g\ell\Delta|\frac{\rd g}{\ell^3}\\
&\leq \int_G|(\partial\Lambda_L)_{R\ell}\cap g\ell\Delta|\frac{\rd g}{\ell^3}=|(\partial\Lambda_L)_{R\ell}|\leq CL^2\ell
\end{align*}
by~\eqref{eq:GS_volume} with $\Omega$ replaced by $(\partial\Lambda_L)_{C\ell}$.
Similarly, we have
$$
0\leq |\Lambda_L|-\int_{I_{L,\ell}}\rd g=\int_{B_{L,\ell}}|\Lambda_L\cap g\ell\Delta|\frac{\rd g}{\ell^3}\leq\int_{B_{L,\ell}}\rd g\leq CL^2\ell.
$$
Using that $\cE_{\Lambda_L\cap g\ell\Delta}\geq0$ for the boundary tetrahedra, we can thus estimate the integral over $G$ by
\begin{align*}
&\int_G \left( \cE_{\Lambda_L\cap g\ell\Delta}[\Omega\cap g\ell\Delta]-\mu_*|\Omega\cap g\ell\Delta|\right)\frac{\rd g}{\ell^3}\\
&\qquad\qquad\geq \frac{F_{\ell\Delta}(\rho)}{\ell^3}\int_{I_{L,\ell}} \rd g-\mu_*\int_{B_{L,\ell}} \rd g\\
&\qquad\qquad\geq \frac{F_{\ell\Delta}(\rho)}{\ell^3}L^3-C\left(\mu_*+\frac{|F_{\ell\Delta}(\rho)|}{\ell^3}\right)L^2\ell.
\end{align*}
Inserting this in~\eqref{eq:GS} we obtain after minimizing over $\Omega$ with the global neutrality constraint $|\Omega|=\rho|\Lambda_L|$
\begin{equation}
		\frac{E_{\Lambda_L}(\rho)}{L^3}-\mu_*\,\rho\geq \frac{F_{\ell\Delta}(\rho)}{\ell^3}-C\left(\mu_*+\frac{|F_{\ell\Delta}(\rho)|}{\ell^3}\right)\frac{\ell}L - C \frac\rho\ell.\label{eq:GS_final}
	\end{equation}
Finally, taking the limit $L\to\ii$ using Theorem~\ref{thm:thermo_limit} we obtain the claimed inequality~\eqref{eq:lower_simplex}.
\end{proof}

In the proof we used the following localization lemma concerning the perimeter.

\begin{lemma}[Localization of the perimeter]
	Given a tetrahedron $\Delta\subset\R^3$ of unit volume, we have
	$$
	\Per(\Omega) = \int_G \Per(\Omega\cap g\ell\Delta)\,\frac{\rd g}{\ell^3} - \Per(\Delta) \frac{|\Omega|}{\ell}
	$$
	for any $\ell>0$ and for any Borel set $\Omega\subset\R^3$ of finite measure and finite perimeter.
\end{lemma}

\begin{proof}
By scaling it suffices to prove this for $\ell=1$. We may also assume that $\Omega$ is sufficiently smooth (for instance Lipschitz). The general case then follows either by approximation or else by working with the essential boundary in the sense of De Giorgi.

We denote the faces of $\Delta$ by $\Gamma_j$, $j=0,1,2,3$. Then
	$$
	\Per(\Omega\cap g\Delta) = \mathcal H^2(\partial\Omega \cap g\Delta) + \sum_{j=0}^3 \mathcal H^2(\Omega\cap g\Gamma_j) \,.
	$$
	By Fubini's theorem, we obtain
	$$
	\int_G \mathcal H^2(\partial\Omega \cap g\Delta)\,\frac{\rd g}{\ell^3} = \Per \Omega \,.
	$$
	Thus, our goal is to compute the term involving $\mathcal H^2(\Omega\cap g\Gamma_j)$. We first perform the integral over translations. Let $H_j$ be the affine plane spanned by $\Gamma_j$. The integral over translations can be split into the integral with respect to translations parallel to $H_j$ and those orthogonal to $H_j$. Integrating $\mathcal H^2(\Omega\cap g\Gamma_j)$ over translations parallel to $H_j$, we obtain $\mathcal H^2(\Gamma_j) \ \mathcal H^2(\Omega\cap H_j)$. Integrating this resulting quantity over translation orthogonal to $H_j$ we obtain, by Fubini's theorem, $\mathcal H^2(\Gamma_j)\ |\Omega|$. Finally, we have to integrate with respect to rotations, but since the corresponding measure is a probability measure, this does not change the value $\mathcal H^2(\Gamma_j)\ |\Omega|$. This proves the claimed identity with constant $\sum_{j=0}^3\mathcal H^2(\Gamma_j)=\Per(\Delta)$.
\end{proof}

Since our goal is to derive the term of order $\rho^{4/3}$, it is natural to take $\ell=A\rho^{-1/3}$ where $A$ is a fixed large constant that we will take to infinity \textbf{after we have taken the limit $\rho\to0$}. The bound~\eqref{eq:lower_simplex} can then be written in the form
\begin{equation}
\frac{e(\rho)-\mu_*\rho}{\rho^{4/3}}\geq \frac{F_{\ell\Delta}(\rho)}{\rho^{1/3}A^3}-\frac{C}{A}.
\label{eq:lower_simplex_bis}
\end{equation}

\subsection*{Step 2. Ultra-dilute limit}
Starting from~\eqref{eq:lower_simplex_bis} we need to prove the convergence of $ \rho^{-1/3}F_{\ell\Delta}(\rho)$ to the Jellium problem in finite volume. This is stated in the following proposition.

\begin{proposition}[Ultra-dilute limit]\label{prop:limit_finite_charge}
Let $\Delta$ a fixed tetrahedron of volume $1$. Then we have for $\ell=A\rho^{-1/3}$ with $A$ a fixed large enough constant
\begin{multline}
\lim_{\rho\to0}\rho^{-\frac13}F_{\ell\Delta}(\rho)=\min_{n\geq0}\min_{x_1,...,x_n\in A\Delta}\bigg\{\sum_{1\leq j<k\leq n}\frac{m_*^2}{|x_j-x_k|}\\
- m_*\sum_{j=1}^n\int_{A\Delta}\frac{\dy}{|x_j-y|}+\frac{1}{2}\iint_{A\Delta\times A\Delta}\frac{\dx\,\dy}{|x-y|}\bigg\}.
\label{eq:to_be_proved}
\end{multline}
\end{proposition}

The right side is the Jellium energy for particles with a charge $m_*$, in a fixed tetrahedron of volume $A^3$ and a uniform background of density 1. It is grand-canonical because we optimize over the number of particles $n$ instead of improsing the neutrality condition $m_*n=A^3$. It is however known that the optimal $n$ satisfies $n=A^3/m_*+o(A^2)$. Fluctuations about the neutral case are negligible compared to the surface, otherwise the resulting energy is modified to the order of the volume (see~\cite{LieLeb-72,LieNar-75,GraSch-95b} and \cite[Cor.~50]{Lewin-22}).

We call~\eqref{eq:to_be_proved} the ``ultra-dilute'' limit because we expect that the solution to the left side involves finitely many droplets (of the order of $A^3$) located at distances of order $\rho^{-1/3}$ in the very large domain $\ell\Delta$ of side length $\ell=A\rho^{-1/3}$. A similar limit was considered for the canonical problem in~\cite{ChoPel-10,KnuMurNov-16,EmmFraKon-20}. After rescaling everything by $\rho^{1/3}$ the droplets look like points and we find the right side.

To prove Proposition~\ref{prop:limit_finite_charge}, we will need the following auxilliary lemma, which provides an estimate on the volume of minimizers for the grand-canonical liquid drop problem with background in~\eqref{eq:def_F}.

\begin{lemma}[Estimate on the mass]\label{lem:estim_mass}
Let $\Omega\subset \Lambda$ be so that $\cE_\Lambda[\rho,\Omega]\leq \mu_*\,|\Omega|$. Then we have
\begin{equation}
|\Omega|\leq 8+16\pi\rho\ {\rm diam}(\Lambda)^3.
\label{eq:mass_bound}
\end{equation}
\end{lemma}

For $\rho=0$ we just recover the result from~\cite{FraKilNam-16}, of which the lemma is an easy extension. When $\rho>0$ the inequality~\eqref{eq:mass_bound} says that $|\Omega|$ is at most of the same order as the background charge $\rho\,|\Lambda|\approx \rho\ {\rm diam}(\Lambda)^3$.

\begin{proof}
Our proof follows the method in~\cite{FraKilNam-16}. For any $\nu\in\bS^2$ and $t\in\R$, we define the half spaces and the plane
$$\cH^-_{\nu,t}:=\{x\ :\ x\cdot\nu\leq t\},\qquad \cH^+_{\nu,t}:=\{x\ :\ x\cdot\nu> t\},\qquad \cP_{\nu,t}=\{x\ :\ x\cdot\nu=t\}.$$
Upon translating everything we can assume that $\Lambda\subset B(0,R)$ with $R=\text{diam}(\Lambda)$ (in fact, by Jung's theorem~\cite[Thm.~2.10.41]{Federer-96} we could take $R=\text{diam}(\Lambda)\sqrt{3/8}$ which would slightly improve the constant $16\pi$ in the statement of the lemma). Then we know that $\cH^\pm_{\nu,\pm t}\cap\Lambda=\emptyset$ for all $\pm t\geq R$. For every $t\in (-R,R)$, we split $\Omega$ into
$$\Omega=\Omega_{\nu,t}^+\cup \Omega_{\nu,t}^-,\qquad \Omega_{\nu,t}^\pm=\Omega\cap \cH_{\nu,t}^\pm.$$
Using $\cE_\Lambda[0,\Omega^\pm_{\nu,t}]=I[\Omega^\pm_{\nu,t}]\geq\mu_*|\Omega^\pm_{\nu,t}|$ by definition of $\mu_*$ we obtain for almost every $t$ and $\nu$
\begin{align*}
\cE_\Lambda[\rho,\Omega]&=\cE_\Lambda[0,\Omega^+_{\nu,t}]+\cE_\Lambda[0,\Omega^-_{\nu,t}]-2\Per(\Omega\cap\cP_{\nu,t})\\
&\qquad -2\rho D(\1_\Omega,\1_{\Lambda})+\rho^2D(\1_{\Lambda})+2D(\1_{\Omega^+_{\nu,t}},\1_{\Omega^-_{\nu,t}})\\
&\geq \mu_*|\Omega|-2\Per(\Omega\cap\cP_{\nu,t})+2D(\1_{\Omega^+_{\nu,t}},\1_{\Omega^-_{\nu,t}})-2\pi\rho R^2|\Omega|.
\end{align*}
In the last line we have used that
$$\1_\Lambda\ast|\cdot |^{-1}\leq \1_{B(0,R)}\ast|\cdot |^{-1}=\int_{B(0,R)}\frac{\dy}{\max(|x|,|y|)}\leq 2\pi R^2$$
by Newton's theorem as in~\eqref{eq:Newton}.
Next we integrate over $\nu\in\bS^2$ and $t\in(-R,R)$, which gives as in~\cite{FraKilNam-16}
$$\frac1{4\pi}\int_{\bS^2}\int_{-R}^R \Per(\Omega\cap\cP_{\nu,t})\,\rd t\,\rd\nu=|\Omega|$$
and
$$\frac1{4\pi}\int_{\bS^2}\int_{-R}^R D(\1_{\Omega^+_{\nu,t}},\1_{\Omega^-_{\nu,t}})\,\rd t\,\rd \nu=\frac{|\Omega|^2}{8}.$$
Hence we obtain
\begin{equation*}
\frac{|\Omega|^2}4-2|\Omega| \leq 4\pi\rho R^3|\Omega|+2R (\cE_\Lambda[\rho,\Omega]-\mu_*|\Omega|)\leq 4\pi\rho R^3|\Omega|
\end{equation*}
and the result follows after dividing by $|\Omega|$.
\end{proof}

We are now able to provide the

\begin{proof}[Proof of Proposition~\ref{prop:limit_finite_charge}]
The upper bound
\begin{multline}
\limsup_{\rho\to0}\rho^{-\frac13}F_{\ell\Delta}(\rho)\leq \min_{n\geq0}\min_{x_1,...,x_n\in A\Delta}\bigg\{\sum_{1\leq j<k\leq n}\frac{m_*^2}{|x_j-x_k|}\\
- m_*\sum_{j=1}^n\int_{A\Delta}\frac{\dy}{|x_j-y|}+\frac{1}{2}\iint_{A\Delta\times A\Delta}\frac{\dx\,\dy}{|x-y|}\bigg\}
\label{eq:upper-bound-ultra-dilute}
\end{multline}
is shown by placing $n=A^3/m_*+o(A^2)$ copies of a minimizer $\Omega_*$ at points solving the right side, rescaled by a factor $\rho^{1/3}$, and then averaging over rotations of the droplets to simplify the computation of the Coulomb interaction. This works exactly as in Section~\ref{sec:proof_upper_bound}. We omit the details since Theorem~\ref{thm:low_density} only relies on the lower bound.

For the lower bound, we take any sequence $\rho_n\to0$ and call $\Omega_n$ a minimizer for $F_{\ell_n\Delta}(\rho_n)$, with $\ell_n=A\rho_n^{-1/3}$. We know that the right side of~\eqref{eq:upper-bound-ultra-dilute} is non-positive.
Let us recall this argument. For any fixed $n$, the second minimum in~\eqref{eq:upper-bound-ultra-dilute} over the $x_j$'s can be bounded from above by its average, which equals
$$J_n:=\left(\frac{n(n-1)}{2A^6}m_*^2-\frac{m_*n}{A^3}+\frac12\right)\iint_{A\Delta\times A\Delta}\frac{\dx\,\dy}{|x-y|}.$$
Writing $A^3/m_*=N+t=(1-t)N+t(N+1)$ with $N\in\N_0$ and $0\leq t<1$ we can then bound the first minimum by the convex combination
$$(1-t)J_N+tJ_{N+1}=-\frac{N+t^2}{2A^6}m_*^2\iint_{A\Delta\times A\Delta}\frac{\dx\,\dy}{|x-y|}<0.$$
From the upper bound~\eqref{eq:upper-bound-ultra-dilute} we conclude that $F_{\ell_n\Delta}(\rho_n)<0$ for $\rho_n$ small enough. Therefore, by Lemma~\ref{lem:estim_mass} we obtain that
$$|\Omega_n|\leq 8+16\pi \rho_n\ell_n^3{\rm diam}(\Delta)^3=8+16\pi A^3{\rm diam}(\Delta)^3$$
is bounded. Recall that $A$ is fixed in the whole argument. Next, from $F_{\ell_n\Delta}(\rho_n)<0$ and discarding the self-interaction of the background, we obtain
$$I[\Omega_n]-\mu_*|\Omega_n|\leq \rho_n\int_{\Omega_n}\int_{\ell_n\Delta}\frac{\dx\,\dy}{|x-y|}\leq C\rho_n\ell_n^2|\Omega_n|\leq \frac{CA^3}{\ell_n}.$$
We have again used Newton's theorem to estimate the interaction with the background as we did already in~\eqref{eq:Newton}.
We are thus exactly in the setting of Theorem~\ref{thm:decomposition} and Corollary~\ref{cor:repulsion_ell_n}. We know that we can decompose
$$\Omega_n=\bigcup_{j=1}^K\Omega_n^{(j)}\cup \Omega_n^o,\qquad |\Omega_n^o|+\Per(\Omega_n^o)\to0$$
with $\Omega_n^{(j)}=\Omega_n\cap B(a_n^{(j)},R)$, $\Omega_n^{(j)}-a_n^{(j)}\to \Omega_*^{(j)}$ globally (that is, in $L^1$ for the corresponding characteristic functions) which are minimizers for $\mu_*$ of mass $m_*^{(j)}:=|\Omega_*^{(j)}|\in[m_{**},m_*]$ and with the centers satisfying
$a_n^{(j)}/\ell_n\to y_j$, for $K$ distinct points in $\overline{\Delta}$. From~\eqref{eq:repulsion2} we have
$$I[\Omega_n]\geq \mu_*|\Omega_n|+\frac1{\ell_n}\sum_{1\leq j< k\leq K}\frac{m_*^{(j)}m_*^{(k)}+o(1)}{|y_j-y_k|}.$$
Here, it is understood that the second term is absent for $K=1$.
For the interaction with the background we use~\eqref{eq:Newton} again and obtain
\begin{multline*}
\rho_n\int_{\Omega_n^{(j)}}\int_{\ell_n\Delta}\frac{\dx\,\dy}{|x-y|}\\
\leq \rho_n\int_{a_n^{(j)}+\Omega_*^{(j)}}\int_{\ell_n\Delta}\frac{\dx\,\dy}{|x-y|}+\frac{CA^3}{\ell_n}\int\Big|\1_{(\Omega_n^{(j)}-a_n^{(j)})}-\1_{\Omega_*^{(j)}}\Big|.
\end{multline*}
The last term is a $o(1/\ell_n)$. Multiplying by $\ell_n$ we find after a change of variables that the first term behaves as
\begin{align*}
\ell_n\rho_n\int_{a_n^{(j)}+\Omega_*^{(j)}}\int_{\ell_n\Delta}\frac{\dx\,\dy}{|x-y|}&=A^3\ell_n^3\int_{\frac{a_n^{(j)}+\Omega_*^{(j)}}{\ell_n}}\1_{\Delta}\ast|\cdot|^{-1}(y)\,\dy\\
&\to A^3m_*^{(j)}\1_{\Delta}\ast|\cdot|^{-1}(y_j)
\end{align*}
where we used that $\1_{\Delta}\ast|\cdot|^{-1}$ is a continuous function and that
$$\ell_n^3\1_{\frac{a_n^{(j)}+\Omega_*^{(j)}}{\ell_n}}\wto m_*^{(j)}\delta_{y_j}$$
in the sense of measures. For the interaction with $\Omega_n^o$ we have by~\eqref{eq:Newton}
$$\rho_n\int_{\Omega_n^o}\int_{\ell_n\Delta}\frac{\dx\,\dy}{|x-y|}\leq C\rho_n\ell_n^2|\Omega_n^o|=CA^3\frac{|\Omega_n^o|}{\ell_n}=o\left(\frac1{\ell_n}\right).$$
Finally, we have for the self-interaction of the background
$$\frac{\rho_n^2}{2}\iint_{(\ell_n\Delta)^2}\frac{\dx\,\dy}{|x-y|}=\frac{A^6}{2\ell_n}\iint_{\Delta^2}\frac{\dx\,\dy}{|x-y|}.$$
Our conclusion is that
\begin{multline*}
\liminf_{n\to\ii}\ell_n\,F_{\ell_n\Delta}(\rho_n)\geq\sum_{1\leq j< k\leq K}\frac{m_*^{(j)}m_*^{(k)}}{|y_j-y_k|}-A^3\sum_{j=1}^Km_*^{(j)}\int_{\Delta}\frac{\dy}{|y_j-y|}\\
+\frac{A^6}{2}\iint_{\Delta^2}\frac{\dx\,\dy}{|x-y|}.
\end{multline*}
Next we remark as in~\cite{DysLen-68,DauLie-83} that the right side is linear with respect to each $m_*^{(j)}$, when all the other ones are fixed. We can thus replace all the $m_*^{(j)}$ by either $0$ or the maximal value $m_*$ and decrease the right side. When they are replaced by $0$ this is equivalent to discarding the corresponding points $y_j$, that is, to decrease $K$. Relabelling the points $y_j$ so that the $K'$ first end up having mass $m_*$, we find that
\begin{multline*}
\liminf_{n\to\ii}\ell_n\,F_{\ell_n\Delta}(\rho_n)\geq m_*^2\sum_{1\leq j< k\leq K'}\frac{1}{|y_j-y_k|}-A^3m_*\sum_{j=1}^{K'}\int_{\Delta}\frac{\dy}{|y_j-y|}\\
+\frac{A^6}{2}\iint_{\Delta^2}\frac{\dx\,\dy}{|x-y|}.
\end{multline*}
After rescaling by $A$ this is the claimed result in~\eqref{eq:to_be_proved}.
\end{proof}

\subsection*{Step 3. Conclusion of the proof}
Inserting the limit~\eqref{eq:to_be_proved} in~\eqref{eq:lower_simplex_bis} we find
\begin{multline*}
\liminf_{\rho\to0}\frac{e(\rho)-\mu_*\rho}{\rho^{4/3}}\geq A^{-3}\min_{n\geq0}\min_{x_1,...,x_n\in A\Delta}\bigg\{\sum_{1\leq j<k\leq n}\frac{m_*^2}{|x_j-x_k|}\\
- m_*\sum_{j=1}^n\int_{A\Delta}\frac{\dy}{|x_j-y|}+\frac{1}{2}\iint_{A\Delta\times A\Delta}\frac{\dx\,\dy}{|x-y|}\bigg\}-\frac{C}{A}.
\end{multline*}
Here we can rescale everything by a factor $m_*^{1/3}$ and, after passing to the limit $A\to\ii$, we obtain the Jellium energy
$$\liminf_{\rho\to0}\frac{e(\rho)-\mu_*\rho}{\rho^{4/3}}\geq m_*^{2/3}e_\text{Jel}$$
as we wanted. The grand-canonical Jellium problem has the same thermodynamic limit as its canonical counterpart in~\eqref{eq:Jellium} by~\cite[Thm.~24]{Lewin-22}. This concludes the proof of Theorem~\ref{thm:low_density}.\qed

\appendix\section{Proof of Theorem~\ref{thm:thermo_limit} on the thermodynamic limit}\label{app:proof_thermo-limit}

In this appendix we describe how to prove the existence of the thermodynamic limit stated in~\eqref{eq:thermo_limit}, using the method introduced by Lieb, Lebowitz and Narnhofer in~\cite{LieLeb-72,LieNar-75}. It is also presented  in~\cite[Chap. 14]{LieSei-09}.

The proof is, as usual, divided into three steps. The first step consists of establishing the existence of the limit for a special sequence of balls $B_j$ growing exponentially fast. This is done by packing any $B_K$ by a certain disjoint union of copies of the smaller balls $B_0,...,B_{K-1}$ in such a way that the remaining space has a volume negligible compared with the large volume $|B_K|$. The existence of such packings was proved first in~\cite{LieLeb-72} and has become known as a ``swiss cheese theorem''. In the next two steps we use this special sequence to obtain the limit for any sequence of sets $\Lambda_n\nearrow\R^3$, satisfying suitable regularity conditions. Upper and lower bounds are proved using a similar swiss cheese for the inside or the outside of $\Lambda_n$.

\subsection{Construction of a quadrupole layer}
We start by presenting a technical result that will be used in each of the three steps. It will allows us to deal with the neighborhood of the boundary of the sets we will manipulate. Recall that for the slowly-decaying Coulomb potential, a small charge disturbance can have a large effect far away. When constructing trial sets, it is therefore important to ensure that the charges locally balance well. Local neutrality is not enough and in the following proposition we will achieve an additional vanishing dipole condition. Handling this situation for the liquid drop model is harder than what was previously done in~\cite{LieLeb-72,LieNar-75} and this is our main contribution in this appendix.

\begin{proposition}[Quadrupole layer]\label{prop:dipole_layer}
Let $\rho\leq1/2$. Let $\Lambda\subset\R^3$ be an $(r,L)$--Lipschitz open set with compact boundary. Let $C_z=\eps z+(-\eps/2,\eps/2)^3$ be a tiling of cubes of size $\eps>0$, with $z\in\Z^3$. Define the outer cubic boundary layer of $\Lambda$ by
$$\cB:=\bigcup_{\rd(C_z,\partial\Lambda)\leq \eps}\overline{C_z\setminus \Lambda}.$$
If $\eps\leq r/C$, then we can write $\cB$ as a disjoint (up to a negligible set) union
$$\cB=\bigcup_{\alpha=1}^{A} \overline{\Lambda_\alpha},\qquad A\leq C\eps^{-3}\left|\big\{x\in \R^3\setminus \Lambda\ :\ \rd(x,\partial\Lambda)\leq C\eps\big\}\right|.$$
Each set $\Lambda_\alpha$ contains a ball of radius $\eps/C$ and is contained in a ball of radius $C\eps$. It also contains a set $\Omega_\alpha$ satisfying
\begin{equation}
|\Omega_\alpha|=\rho|\Lambda_\alpha|,\qquad \Per(\Omega_\alpha)\leq C\rho^{\frac23}\eps^2,\qquad \int_{\Omega_\alpha}y\,\dy=\rho\int_{\Lambda_\alpha}y\,\dy.
\label{eq:local_dipoles}
\end{equation}
The constant $C$ only depends on the Lipschitz constant $L$.
\end{proposition}

Recall that an open set $\Lambda\subset\R^3$ is $(r,L)$--Lipschitz  if for any $x\in\partial\Lambda$, the set $\Lambda\cap B(x,r)$ is, in an appropriate coordinate system, equal to the epigraph of an $L$--Lipschitz function. The interpretation of the condition~\eqref{eq:local_dipoles} is that the charge distribution $\1_{\Omega_\alpha}-\rho\1_{\Lambda_\alpha}$ has no charge and no dipole moment. It can be thought as a small quadrupole, hence the title of the proposition. Recall that in our situation the quadrupole moment is the $3\times3$ symmetric trace-less matrix
\begin{equation}
Q_\alpha=\int\left(yy^T-\frac{|y|^2}3 I_3\right)(\1_{\Omega_\alpha}-\rho\1_{\Lambda_\alpha})(y)\,\dy.
\label{eq:quadrupole}
\end{equation}
It would be interesting to make the latter vanish as well.

\begin{figure}[t]
\includegraphics[width=8cm]{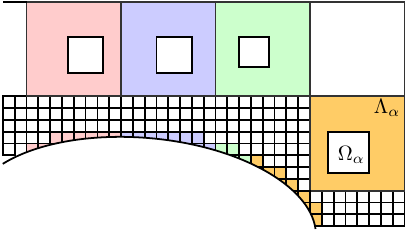}
\caption{Illustration of the proof of Proposition~\ref{prop:dipole_layer}. The small white cubes that are full are the $\Lambda_\alpha$'s of the first type. The large cubes on the upper right are the $\Lambda_\alpha$'s of the second type. The colored large cubes are the ones merged with the partial small cubes of the same color that constitute the $\Lambda_\alpha$'s of the third type. In those cubes, $\Omega_\alpha$ is chosen to be a cube of volume $\rho|\Lambda_\alpha|$, placed at the center of mass of $\Lambda_\alpha$. The latter is at distance $O(\eps/K)$ to the center of the large cube. \label{fig:dipole_layer}}
\end{figure}

\begin{proof}
We call $\cZ$ the set of indices $z$ such that $C_z$ intersects the complement of $\Lambda$ and $\rd(C_z,\partial\Lambda)\leq\eps$, so that $\cB=\cup_{z\in\cZ}\overline{C_z\setminus \Lambda}$. We call $\cZ_i\subset\cZ$ the indices such that the cube $C_z$ is contained in the complement of $\overline\Lambda$ and $\cZ_b:=\cZ\setminus \cZ_i$ the indices so that the cube $C_z$ intersects the boundary $\partial \Lambda$.

Next, we split each cube $C_z$ intersecting the boundary into many smaller cubes of side length $\eps/K$ with $K\in\N$ a large enough integer to be determined later. We denote the associated tiling by
$$C'_{z'}=\frac{\eps}{K}\big(z'+(-1/2,1/2)^3\big),\qquad z'\in\Z^3.$$
Any of the old cubes is a union of $K^3$ small cubes. We introduce the corresponding sets of indices
$$\cZ'_i=\left\{z'\in\Z^d\ :\ C'_{z'}\subset \cup_{z\in\cZ_b}C_z,\quad C'_{z'}\subset \R^3\setminus \overline\Lambda\right\},$$
$$\cZ'_b=\left\{z'\in\Z^d\ :\ C'_{z'}\subset \cup_{z\in\cZ_b}C_z,\quad C'_{z'}\cap \partial \Lambda\neq\emptyset\right\}.$$
Now, to each $z'\in\cZ'_b$ (a \emph{small} cube intersecting the boundary $\partial\Lambda$) we associate a unique index $z\in \cZ_i$ (a \emph{large} cube not intersecting the boundary $\partial\Lambda$). We simply take $z$ such that the distance between $C_z$ and $C'_{z'}$ is minimal (a choice has to be made in case there are several possibilities). The sets $\Lambda_\alpha$ in the statement then consist of three different types. First we have the collection of cubes $C'_{z'}$ with $z'\in\cZ'_i$. In any such cube we take $\Omega_\alpha=(\eps/K) z'+\rho^{1/3}C'_{0}$. The properties~\eqref{eq:local_dipoles} follow immediately. The next type of sets are those large cubes $C_z$ with $z\in\cZ_i$ that are not associated to any small cube $C'_{z'}$ at all. We again pick $\Omega_\alpha=\eps z+\rho^{1/3}C_{0}$ and obtain the desired properties. In these two cases we in fact also have a vanishing quadrupole $Q_\alpha=0$.

Finally we come to the most interesting type of $\Lambda_\alpha$. We define each remaining $\Lambda_\alpha$ as the union of one $C_z$ with $z\in \cZ_i$ together with all the small partial cubes $C'_{z'}\setminus\Lambda$ associated to it. The picture is therefore that we have a large full cube of volume $\eps^3$ together with several small partial cubes of volume $\leq \eps^3/K^3$ (Figure~\ref{fig:dipole_layer}). The $(r,L)$--Lipschitz property implies the following two facts, for $\eps\leq r/C$:
\begin{enumerate}
\item[(a)] The distance of any small cube with its associated large cube is bounded by $C\eps$.
\item[(b)] The number of small cubes associated with a given large cube is bounded by $CK^2$.
\end{enumerate}
The constant $C$ only depends on $L$. The first claim (a) follows for instance from the fact that at any point $x\in\partial\Lambda$ there is a cone of height comparable with $r$ and opening angle bounded from below in terms of $L$, that is completely contained in $\R^3\setminus \Lambda$ (see for instance~\cite[Thm.~1.2.2.2]{Grisvard}). Hence for $\eps\leq r/C$ this cone will contain a full cube $C_z$. For the second claim (b), we use that the $\eta$--thickened boundary inside a ball of radius $r'$ satisfies
$$\left|\big(\partial\Lambda+B(0,\eta)\big)\cap B(x,r')\right|\leq C(r')^2\eta$$
for all $r',\eta\leq r$ and $x\in\partial\Lambda$. Taking $r'$ of the order $\eps$ and $\eta=\sqrt{3}\eps/K$, we obtain (b).

Next we construct the set $\Omega_\alpha$. Our idea is that when $K$ is large enough, the center of mass of $\Lambda_\alpha$ will be very close to that of $C_z$ (that is, $\eps z$). This allows us to place a cube $\Omega_\alpha$ of the appropriate volume fully included in the large cube $C_z\subset\Lambda_\alpha$, about its center of mass, so as to cancel the dipole moment. The center of mass of $\Lambda_\alpha$ is given by
$$\frac1{|\Lambda_\alpha|}\int_{\Lambda_\alpha}y\,\dy=\frac{\eps^3}{\eps^3+|U|}\eps z+\frac1{\eps^3+|U|}\int_U y\,\dy=\eps z+\frac1{\eps^3+|U|}\int_U (y-\eps z)\,\dy$$
where $U:=\Lambda_\alpha\setminus C_z$ contains the partial small cubes. Using (a) we can estimate the shift by
$$\left||\Lambda_\alpha|^{-1}\int_{\Lambda_\alpha}y\,\dy-\eps z\right|\leq \frac1{\eps^3+|U|}\int_U |y-\eps z|\,\dy\leq C\frac{\eps|U|}{\eps^3+|U|}.$$
By (b) we have $|U|\leq C\eps^3/K$. Therefore we obtain
$$\left||\Lambda_\alpha|^{-1}\int_{\Lambda_\alpha}y\,\dy-\eps z\right|\leq C\frac{\eps}{K+1}.$$
This means that the center of mass of $\Lambda_\alpha$ is located at a distance of order $\eps/K$ from the center of mass $\eps z$ of $C_z$. We can now place a cube of the appropriate volume centered at this point
$$\Omega_\alpha:=\frac1{|\Lambda_\alpha|}\int_{\Lambda_\alpha}y\,\dy+\frac{\left(\rho|\Lambda_\alpha|\right)^{\frac13}}{\eps}C_0.$$
This cube has side length
$$\eps\rho^{\frac13}\left(1+\frac{|U|}{\eps^3}\right)^{\frac13}\leq \eps\rho^{\frac13}\left(1+\frac{C}{K}\right)^{\frac13}.$$
Since $\rho\leq1/2$, this set is fully included into $C_z$ for $K$ larger than a constant depending only on $L$. As a conclusion we have constructed the desired sets. By definition each $\Lambda_\alpha$ contains at least one small cube $C'_{z'}$. They are all disjoint and cover $\cB$ which is included in $V=\{x\in \R^3\setminus \Lambda\ :\ \rd(x,\partial\Lambda)\leq C\eps\}$. It follows that there are at most $K^3\eps^{-3}|V|$ such sets.
\end{proof}

\subsection{Limit for a sequence of exponentially growing balls}
We now follow~\cite{LieLeb-72,LieNar-75} and introduce
$$R_j=(1+p)^j(1-\theta^j/2),\qquad n_j=\frac{\gamma^j(1+p)^{3j}}p$$
where
$$p=26,\qquad \gamma=\frac{p}{1+p},\qquad \theta=\frac1{1+p}.$$
Let $B_k:=B(0,R_k)$. By~\cite[Lemma 4.1]{LieLeb-72}, one can pack the ball $B_K$ by a disjoint union of $n_{K-j}$ translated copies of the $B_j$'s with $j=0,...,K-1$, in such a way that the mutual distance between all the balls, including $\partial B_K$, is at least equal to $\frac{p}{2(1+p)}$ (Figure~\ref{fig:packing_ball}). A simple calculation shows that the total occupied volume satisfies
\begin{equation}
\frac{\gamma^K|B_K|}{C}\leq |B_K|-\sum_{j=0}^{K-1}n_{K-j}|B_j|\leq C\gamma^K|B_K|,
\label{eq:filling_balls}
\end{equation}
that is, the volume fraction of the remaining part is exponentially small. A similar estimate provides for the perimeter
\begin{multline}
\frac{\gamma^K|B_K|}{C}\leq \sum_{j=0}^{K-1}n_{K-j}\Per(B_j)=4\pi\sum_{j=0}^{K-1}n_{K-j}R_j^2\leq C\gamma^K|B_K|,\\
\Per(B_K)\leq C\theta^K|B_K|.
\label{eq:estim_perimeter}
\end{multline}
The perimeter of the remaining part is thus comparable to its volume. Everywhere $C$ can be chosen universal.

\begin{figure}[t]
\includegraphics[width=8cm]{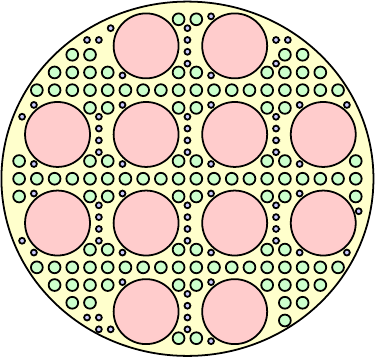}
\caption{Packing the large ball $B_K$ with smaller balls $B_j$ as in~\cite{LieLeb-72,LieNar-75} so that the uncovered part (in yellow in the picture) has an exponentially small fraction of the total volume $|B_K|$.\label{fig:packing_ball}}
\end{figure}

Let $\Omega_j$ be a minimizer for the problem $E_{B_j}(\rho)$ in the ball $B_j=B(0,R_j)$. We call $x_{j,n}$ the centers of the packing balls of radii $R_j$ that are placed in the large ball $B_K$. In each of these balls, we place the trial set $\Omega_{j,n}:=x_{j,n}+U_{j,n}\Omega_j$ with a rotation $U_{j,n}\in SO(3)$. When we average over $U_{j,n}$ the Coulomb potential induced by $\1_{\Omega_{j,n}}-\rho\1_{B(x_{j,n},R_j)}$ vanishes outside of the corresponding ball $B(x_{j,n},R_j)$ by Newton's theorem as in~\eqref{eq:Newton_averaged}. This is due to the neutrality condition $|\Omega_j|=\rho|B_j|$. The Coulomb interaction between this ball and the rest of the system thus vanishes. We do this for all the small balls.

Next we discuss how to deal with the part in $B_K$ that is not covered by the balls:
\begin{equation}
D_K:=B_K\setminus \bigcup_{j=0}^{K-1}\bigcup_{n=1}^{n_{K-j}}B(x_{j,n},R_j).
\label{eq:cheese}
\end{equation}
We have by~\eqref{eq:filling_balls} and~\eqref{eq:estim_perimeter}
\begin{equation}
\frac{\gamma^K|B_K|}C\leq \Per(D_K)\leq C\gamma^K|B_K|,\quad \frac{\gamma^K|B_K|}C\leq |D_K|\leq C\gamma^K|B_K|.
\label{eq:estim_D_K_upper_lower}
\end{equation}
Like in~\cite{LieNar-75} (and unlike in~\cite{LieLeb-72}), we cannot leave $D_K$ empty because the Coulomb energy of the background $\rho\1_{D_K}$ is too large. Whatever $\Omega$ we put in $D_K$ will not interact with the balls, since those generate a Coulomb potential that exactly vanishes outside of the balls, as we have already explained. Hence after taking a minimizing $\Omega\subset D_K$ we obtain the bound
$$E_{B_K}(\rho)\leq \sum_{j=0}^{K-1}n_{K-j}E_{B_j}(\rho)+E_{D_K}(\rho).$$
Using $(1-\theta^j/2)^3\leq (1-\theta^K/2)^3$ since $\theta\leq1$ we get
$$\frac{E_{B_K}(\rho)}{|B_K|}\leq \frac1p\sum_{j=0}^{K-1}\gamma^{K-j}\frac{E_{B_j}(\rho)}{|B_j|}+\frac{E_{D_K}(\rho)}{|B_K|}.$$
We claim that
\begin{equation}
\boxed{E_{D_K}(\rho)\leq C|D_K|\leq C\gamma^K|B_K|.}
\label{eq:estim_cheese}
\end{equation}
(recall that the two terms on the right side are comparable by~\eqref{eq:estim_D_K_upper_lower}).
It follows from~\eqref{eq:estim_cheese} that the limit
$$e(\rho):=\lim_{k\to\ii}\frac{E_{B_K}(\rho)}{|B_K|}$$
exists. This is a consequence of the following lemma, which is essentially contained in~\cite[p.~343]{LieLeb-72} (see also~\cite[p.~264]{LieSei-09}) and whose proof we give for the convenience of the reader.

\begin{lemma}
Let $f_j$ be a sequence of real numbers that is bounded from below. Let $0<\gamma<1$ and assume that there exists a summable sequence $\delta_j\geq0$ such that
$$f_K\leq \frac{1-\gamma}{\gamma}\sum_{j=0}^{K-1}\gamma^{K-j} f_{j}+\delta_K.$$
Then $f_K$ admits a limit as $K\to\ii$.
\end{lemma}

\begin{proof}
Let us introduce $g_K:=f_K-\frac{1-\gamma}\gamma\sum_{j=0}^{K-1}\gamma^{K-j} f_{j}$ which satisfies $g_K\leq \delta_K$. Solving the recursion relation we find
\begin{equation}
f_K=g_K+(1-\gamma)\sum_{j=0}^{K-1}g_j.
\label{eq:f_K_g_K}
\end{equation}
Using that $g_j\leq \delta_j$ which is summable by assumption, we infer that $f_j$ is bounded from above. Inserting this into the definition of $g_K$ we conclude that $g_K$ is bounded from below. The assumption that $f_K$ is bounded from below then implies that the sum $\sum_{j=0}^\ii (g_j)_-$ is finite. Hence by summability of $(g_j)_+\leq\delta_j$ we infer that $g_j$ is summable. Then $g_K\to0$ and thus $f_K$ converges to $(1-\gamma)\sum_{j=0}^\ii g_j$, by~\eqref{eq:f_K_g_K}.
\end{proof}

If therefore only remains to explain how to show~\eqref{eq:estim_cheese}. The construction of a competitor $\Omega$ for $E_{D_K}(\rho)$ is the main difficulty, compared with what was done in~\cite{LieNar-75}. It is harder to screen the background in the liquid drop model because $\1_\Omega$ has a fixed height. Our idea is to split $D_K$ into many small boxes of size $\eps$ and place in each full cube a smaller cube of side length $\rho^{1/3}\eps$. Of course the main difficulty is to deal with the cubes intersecting the boundary of $D_K$ and this is where Proposition~\ref{prop:dipole_layer} becomes useful.

We consider the previous tiling composed of cubes of side length $\eps$, $C_z:=\eps z+(-\eps/2,\eps/2)^3$ with $z\in\Z^3$. We choose $\eps>0$ small enough so as to be able to apply Proposition~\ref{prop:dipole_layer} to $\Lambda=B(x_{j,n},R_j)$ as well as to $\Lambda=\R^3\setminus B(0,R_K)$. We call $\cB_{j,n}$ and $\cB_K$ the corresponding outer cubic layers given by the proposition. We also require that $(1+\sqrt{3})\eps$ is less than half the distance between the balls which, we recall, is bounded from below by $\frac{p}{2(1+p)}$. This is to ensure that the boundary layers are disjoint. The balls are all $(r_0,L_0)$--Lipschitz with $r_0$ and $L_0$ independent of $j,K,n$. Hence $\eps$ is in fact universal here.

Then we can write the set $D_K$ defined previously in~\eqref{eq:cheese} as the disjoint union (up to a negligible set)
$$D_K=\bigcup_{j=0}^{K-1}\bigcup_{n=1}^{n_{K-j}}\cB_{j,n}\cup \cB_K\cup\bigcup_{z\in \cZ_K}C_z$$
where $\cZ_K\subset\Z^3$ contains the full cubes $C_z\subset D_K$ that do not appear in the boundary layers. In each $C_z$ with $z\in\cZ_K$ we place the same set $\Omega=\eps z+\rho^{1/3}C_{0}$ as we did before in the construction of the quadrupole layer in Proposition~\ref{prop:dipole_layer}. We decompose the outer cubic layers into the subsets $\Lambda_\alpha$ given by Proposition~\ref{prop:dipole_layer}. To simplify our notation we now write
$$D_K=\bigcup_\alpha\Lambda_\alpha$$
where $\Lambda_\alpha$ correspond either to the sets in a boundary layer, or to a full cube $C_z$ with $z\in\cZ_K$. The corresponding competitor is
$$\Omega=\bigcup_\alpha\Omega_\alpha$$
where each $\Omega_\alpha\subset\Lambda_\alpha$ satisfies~\eqref{eq:local_dipoles}. We denote by $\mu_\alpha:=\1_{\Omega_\alpha}-\rho\1_{\Lambda_\alpha}$ the corresponding measure. From the statement of Proposition~\ref{prop:dipole_layer}, the number of sets $\Lambda_\alpha$ coming from one given $\cB_{j,n}$ is bounded by $CR_j^2=C\Per(B_j)$ (recall that $\eps$ is universal). Similarly for $\cB_K$. Hence by~\eqref{eq:estim_perimeter}, the total number of such $\Lambda_\alpha$'s can be estimated by $C\gamma^K|B_K|$. Similarly, the number of cubes in $\cZ_K$ is bounded by $C|D_K|\leq C\gamma^K|B_K|$. Finally, the perimeter of $\Omega$ can be estimated by the number of components, that is again bounded above by $C\gamma^K|B_K|$. The Coulomb energy is equal to
$$\sum_{\alpha,\alpha'}D(\mu_\alpha,\mu_{\alpha'}).$$
Since the $\mu_\alpha$ are neutral, have no dipole and are contained in a set of diameter less that $C\eps$, we have by ~\cite[Lemma~2.1]{AnaLew-20}
$$\left|D(\mu_\alpha,\mu_{\alpha'})\right|\leq \frac{C}{1+\rd(\Lambda_\alpha,\Lambda_{\alpha'})^5}.$$
The series is therefore convergent for each $\alpha$ and the full double sum can be estimated by the number of sets $\Lambda_\alpha$'s, that is, $C\gamma^K|B_K|$. We have thus proved the claimed inequality~\eqref{eq:estim_cheese}. This concludes the proof of the convergence of the energy for the special sequence $B_j$ of exponentially growing balls.

\subsection{Upper bound for any regular set}\label{sec:upper-bound_thermo}
Next we turn to the upper bound for an arbitrary sequence of sets $\Lambda_n\nearrow \R^3$ satisfying the conditions of Theorem~\ref{thm:thermo_limit}. The idea is to again pack $\Lambda_n$ with many balls congruent to $B_0,...,B_{K-1}$ for some fixed $K$. To this end, it is easier to start packing the large cube $C'_0=|B_K|^{1/3}(-1/2,1/2)^3$ of the same volume as $B_K$ and then use this cube to pack $\Lambda_n$. By~\cite[Lemma~4.1]{LieLeb-72} the cube $C'_0$ can be packed by the same number  $n_{K-j}$ of copies of the ball $B_j$, with the same estimate on the distance between the boundaries.

We thus consider the new tiling $C'_z:=h_Kz+C'_0$ with $h_K=|B_K|^{1/3}$ and $z\in\Z^3$. We let $\cZ'_i$ denote the indices $z\in\Z^3$ so that $C'_z\subset\Lambda_n$ and $\rd(C'_z,\partial\Lambda_n)>h_K$. Similarly we denote by $\cZ'_b$ the set of the indices $z$ such that $C'_z$ intersects $\Lambda_n$ and $\rd(C'_z,\partial\Lambda_n)\leq h_K$. Hence, up to a negligible set,
$$\Lambda_n=\bigcup_{z\in\cZ'_i}C_z'\cup\bigcup_{z\in\cZ'_b}(\Lambda_n\cap C'_z).$$
In all of the cubes $C'_z$ with $z\in\cZ_i'$, we place the packing of balls mentioned previously. We leave the other cubes $C'_z$ with $z\in\cZ'_b$ completely empty (see Figure~\ref{fig:packing-upper-lower}, left). The number of such cubes can be estimated by
$$\#\cZ'_b\leq |B_K|^{-1}\left|\left\{x\in\Lambda_n\ :\ \rd(x,\partial\Lambda_n)\leq (\sqrt3+1)h_K\right\}\right|\leq C\ell_n^2h_K$$
due to the Fisher regularity condition~\eqref{eq:Fisher} with $h=h_K$ (fixed). Similarly, the number of cubes inside satisfies
$$\#\cZ'_i= |B_K|^{-1}|\Lambda_n|+O(\ell_n^2h_K).$$
We let $D_n$ denote the complement of the balls in $\Lambda_n$ and we proceed with the same construction as before, with a smaller lattice of side length $\eps$. We get a quadrupole layer for $\R^3\setminus\Lambda_n$ using the regularity of its boundary (independent of $n$) as well as for all the balls contained in the large cubes inside. We assume $\eps$ is small enough so as to be able to apply Proposition~\ref{prop:dipole_layer} (now $\eps$ depends on the regularity assumption on $\Lambda_n$). Like in the previous section we find
$$E_{D_n}(\rho)\leq C|D_n|.$$
In any given large cube $C'_z\subset\Lambda_n$ with $z\in\cZ'_i$, we have as before $|D_n\cap C'_z|\leq C\gamma^K|B_K|$ so that after summing we obtain $|D_n|\leq C\gamma^K|\Lambda_n|+C\ell_n^2h_K$. Hence we have
$$E_{D_n}(\rho)\leq C\gamma^K|\Lambda_n|+C\ell_n^2h_K.$$
The number of balls of a type $B_j$ is $n_{K-j}\#\cZ'_i$ hence we get from the same arguments as in the previous subsection
\begin{align}
E_{\Lambda_n}(\rho)&\leq \#\cZ'_i\sum_{j=0}^{K-1}n_{K-j}E_{B_j}(\rho)+E_{D_n}(\rho)\nn\\
&\leq \#\cZ'_i\sum_{j=0}^{K-1}n_{K-j}E_{B_j}(\rho)+C\gamma^K|\Lambda_n|+C\ell_n^2h_K.\label{eq:upper_bound_thermo}
\end{align}
After taking the limit $n\to\ii$ at fixed $K$, we obtain
$$\limsup_{n\to\ii}\frac{E_{\Lambda_n}(\rho)}{|\Lambda_n|}\leq \frac1{|B_K|}\sum_{j=0}^{K-1}n_{K-j}E_\rho(B_j)+C\gamma^K.$$
Taking finally $K\to\ii$, using the existence of the limit for balls proved in the first step, we obtain the claimed upper bound
$$\limsup_{n\to\ii}\frac{E_{\Lambda_n}(\rho)}{|\Lambda_n|}\leq e(\rho).$$

\begin{figure}[t]
\includegraphics[width=10cm]{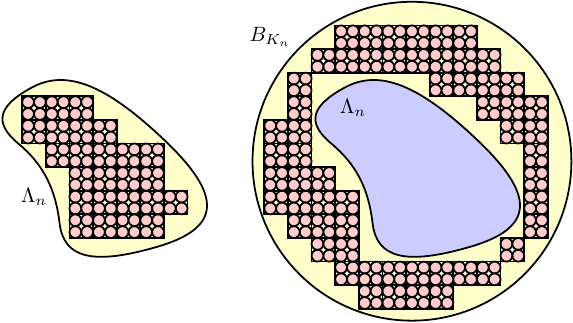}
\caption{Idea of the proof of the thermodynamic limit for an arbitrary sequence $\Lambda_n\nearrow\R^3$, following~\cite{LieLeb-72,LieNar-75}. To get an upper bound on $E_{\Lambda_n}(\rho)$ we pack the domain $\Lambda_n$ by many small balls, using a cubic tiling as guiding principle (only the largest balls of the type $B_{K-1}$ are displayed in the picture). To get a lower bound on $E_{\Lambda_n}(\rho)$ we pack the complement of $\Lambda_n$ in a large ball $B_{K_n}$ of comparable volume.\label{fig:packing-upper-lower}}
\end{figure}

\subsection{Lower bound for any regular set.}
To prove the similar lower bound, we use that $\Lambda_n$ is included in a ball $B(0,\ell_n)$ of volume comparable with $|\Lambda_n|$, by~\eqref{eq:ass_balls_Lambda_n}. We let $K_n$ denote the unique integer such that
$$R_{K_n-1}<\ell_n\leq R_{K_n}.$$
Since $R_{K_n-1}=R_{K_n}(1+p)^{-1}\frac{1-\theta^{K_n-1}/2}{1-\theta^{K_n}/2}$ we conclude that
$$C^{-1}|B_{K_n}|\leq |\Lambda_n|\leq |B_{K_n}|$$
for some constant $C$. In other words the large ball $B_{K_n}$ has a volume of the same order as $B(0,\ell_n)$ and $\Lambda_n$.

Next we pack the complement $B_{K_n}\setminus \Lambda_n$ in the same way as we did in the previous step, using the tiling of cubes $C'_z$ of volume $|B_K|$ and the inside packing of balls (see Figure~\ref{fig:packing-upper-lower}, right).
One difference is that $K$ will actually depend on $n$ and cannot stay fixed, as we did previously. We will need $1\ll K\ll K_n$ and the exact dependence in $n$ will be determined at the end of the argument. As before we define $D_n$ to be the complement of the balls inside $B_{K_n}\setminus \Lambda_n$. We use a tiling of side length $\eps$ (depending on the regularity of $\Lambda_n$) and Proposition~\ref{prop:dipole_layer} to decompose $D_n=\cup_\alpha\Lambda_\alpha$ as before, with trial sets $\Omega_\alpha$ satisfying the no-dipole property~\eqref{eq:local_dipoles}.

In order to get the exact energy $E_{\Lambda_n}(\rho)$, we are forced to take a minimizer $\Omega_n\subset\Lambda_n$. Hence the Coulomb interaction between $\Lambda_n$ and $D_n$ is not suppressed and we obtain the upper bound
\begin{multline}
E_{B_{K_n}}(\rho)\leq \#\cZ'_i\sum_{j=0}^{K-1}n_{K-j}E_{B_j}(\rho)+E_{\Lambda_n}(\rho)+E_{D_n}(\rho)\\
+2\sum_\alpha D\Big(\1_{\Omega_\alpha}-\rho\1_{\Lambda_\alpha},\1_{\Omega_n}-\rho\1_{\Lambda_n}\Big).
\label{eq:lower_bound_energy}
\end{multline}
Here $\cZ'_i$ denotes the sets of indices $z\in\Z^3$ such that $C'_z$ is included in $B_{K_n}\setminus \Lambda_n$ at distance $> h_K$ from its boundary.

The number of $\Lambda_\alpha$'s can be estimated by a constant (depending on $\eps$) times
\begin{multline}
\gamma^K\#\cZ_i'|B_K|+\left|\partial\Lambda_n+B\big(0,(1+\sqrt3) h_K\big)\right|\\+\left|\{R_{K_n}-(1+\sqrt3) h_K\leq|x|\leq R_{K_n}\}\right|\\
\leq C\ell_n^3\left(\gamma^K+\frac{R_K}{R_{K_n}}\right)\leq C\ell_n^3\left(\gamma^K+\gamma^{K_n-K}\right),
\label{eq:nb_quadrupoles}
\end{multline}
where the three terms on the first line come from the small balls in the cubes $C'_z$ with $z\in\cZ_i'$, and from the parts close to $\partial\Lambda_n$ and $\partial B_{K_n}$, respectively. Recall that $h_K=|B_K|^{1/3}=(4 \pi/3)^{1/3} R_K$.  In the first inequality we used the Fisher regularity assumption~\eqref{eq:Fisher} on the boundary of $\Lambda_n$. Arguing as in the previous subsection, we obtain
$$E_{D_n}(\rho)\leq C\ell_n^3\left(\gamma^K+\gamma^{K_n-K}\right).$$

The next step is to handle the Coulomb interaction between the quadrupoles and $\Lambda_n$ in~\eqref{eq:lower_bound_energy}. We use that the Coulomb potential generated by a quadrupole decays like $1/|x|^3$~\cite[Lemma~2.1]{AnaLew-20}, that is,
$$\left|(\1_{\Omega_\alpha}-\rho\1_{\Lambda_\alpha})\ast \frac1{|\cdot|}(x)\right|\leq\frac{C}{1+|x-z|^3}$$
where $z$ is the center of the main cube $C_z$ contained in $\Lambda_\alpha$. Hence the interaction between one quadrupole and $\Lambda_n$ can be estimated by
\begin{align}
\left|D\left(\1_{\Omega_n}-\rho\1_{\Lambda_n},\1_{\Omega_\alpha}-\rho\1_{\Lambda_\alpha}\right)\right|&\leq C\int_{\Lambda_n}\frac{\dx}{1+|x-z|^3}\nn\\
&\leq C\int_{B_{K_n}}\frac{\dx}{1+|x|^3}\nn\\
&\leq C\log(R_{K_n})\leq CK_n.\label{eq:estim_interaction_quadrupole_Lambda}
\end{align}
Several of our $\Lambda_\alpha$'s have a vanishing quadrupole. Those generate a potential decaying like $1/|x|^4$ and the interaction with $\Lambda_n$ is then bounded (there is no $K_n$ on the right side). Since the worse terms will anyway come from the quadrupole layers around the balls, for simplicity we do not distinguish the different cases and use the same bound~\eqref{eq:estim_interaction_quadrupole_Lambda} for all the $\Lambda_\alpha$'s. Using the estimate~\eqref{eq:nb_quadrupoles} on the number of quadrupoles, we arrive at
\begin{multline}
\frac{E_{B_{K_n}}(\rho)}{|B_{K_n}|}\leq \left(1-\frac{|\Lambda_n|}{|B_{K_n}|}\right)\frac1p\sum_{j=0}^{K-1}\gamma^{K-j}\frac{E_{B_j}(\rho)}{|B_j|}+\frac{|\Lambda_n|}{|B_{K_n}|}\frac{E_{\Lambda_n}(\rho)}{|\Lambda_n|}\\
+C\left(\gamma^K+\gamma^{K_n-K}\right)K_n.
\end{multline}
Here we see that we cannot take $K$ fixed as we did for the upper bound, since we need that the last term tends to 0. We take for instance $K=\lfloor\kappa\log(K_n)\rfloor$ with $\kappa\geq 2/\log(\gamma^{-1})$. The convergence towards $e(\rho)$ of the terms involving the standard balls $B_j$ and the fact that $|\Lambda_n|/|B_{K_n}|$ is bounded away from 0 then imply
$$\liminf_{n\to\ii}\frac{E_{\Lambda_n}(\rho)}{|\Lambda_n|}\geq e(\rho).$$
This concludes the proof of Theorem~\ref{thm:thermo_limit}.\qed

\subsection{The periodic problem}
That the periodic model admits the same limit $e(\rho)$, as stated in~\eqref{eq:thermo_limit_periodic}, can be shown by following~\cite[Section V]{LewLieSei-19b}. Let us call
$$\cE_L^{\rm per}[\Omega]:=\Per(\Omega)+\frac12\iint_{\Omega^2}G_L(x-y)\dx\,\dy,\qquad E_L^{\rm per}(\rho):=\min_{|\Omega|=\rho L^3}\cE_L^{\rm per}[\Omega]$$
the periodic energy in a cube $\Lambda_L=(-L/2,L/2)^3$. We take a minimizer $\Omega$ with vanishing center of mass. We then repeat this configuration periodically over a growing cube exactly as we did in Section~\ref{sec:proof_upper_bound}. We find
$$e(\rho)\leq \frac{\cE_L^{\rm per}[\Omega]}{L^3}+\frac{M_{\Z^3}}{2L}.$$
Taking $L\to\ii$ gives
$$e(\rho)\leq \liminf_{L\to\ii}\frac{E_L^{\rm per}(\rho)}{L^3}.$$
To obtain the reverse bound, we pack $\Lambda_L$ with balls as explained in Subsection~\ref{sec:upper-bound_thermo} and in~\cite{LewLieSei-19b}. We choose $\eps\in L/\N$ to avoid having any partial cubes at the boundary of $\Lambda_L$ (and thus be consistent with the periodic boundary condition). Using that Newton's theorem applies in the periodic case as well, we find the same bound as~\eqref{eq:upper_bound_thermo} on $E_L^{\rm per}(\rho)$ and this gives
$$e(\rho)\geq \limsup_{L\to\ii}\frac1{L^3}\min_{\Omega}\cE_L^{\rm per}[\Omega].$$
\qed

\printbibliography

\end{document}